\documentclass[noinfoline,10pt,oneside]{article}

\RequirePackage[breaklinks,colorlinks,citecolor=blue,urlcolor=blue]{hyperref}

\usepackage[normalem]{ulem}
\newif\ifsubsection 
\subsectiontrue\subsectionfalse 

\usepackage[utf8]{inputenc}
\usepackage[english]{babel}
\usepackage{amsmath}
\usepackage{amsfonts}
\usepackage{amssymb}
\usepackage{amsthm}
\usepackage{color}
\usepackage{bbm}
\usepackage{graphicx}
\usepackage{caption}
\usepackage{subcaption}
\usepackage{extsizes}
\usepackage{amsthm}
\usepackage{epsfig}
\usepackage{graphicx}
\usepackage{courier}
\usepackage{colortbl}
\usepackage{multirow}
\usepackage[font=footnotesize]{caption}

\usepackage{algorithm}
\usepackage{algpseudocode}
\usepackage{algcompatible}

\makeatletter
\newcommand{\labitem}[2]{%
\def\@itemlabel{\text{#1}}
\item
\def\@currentlabel{#1}\label{#2}}
\makeatother

\makeatletter
\renewcommand{\ALG@name}{}
\makeatother
\newcommand{\eig}{u_1} 
\newcommand{\R}{\mathbb R }

\newcommand{\tune}{T}
\newcommand{\tunej}{T_j}
\newcommand{\spag}{s_j}
\newcommand{\subgp}{\sigma}

\newtheorem{theorem}{Theorem}
\RequirePackage[OT1]{fontenc}
\RequirePackage{amsthm,amsmath}

\usepackage{color}

\usepackage{float}
\usepackage{booktabs}
\usepackage{marvosym,enumitem}
\usepackage{array, xcolor, lipsum, bibentry}

\definecolor{lightgray}{gray}{0.8}
\newcolumntype{L}{>{\raggedleft}p{0.14\textwidth}}
\newcolumntype{R}{p{0.8\textwidth}}

\newcommand{\vertiii}[1]{{\left\vert\kern-0.25ex\left\vert\kern-0.25ex\left\vert #1
    \right\vert\kern-0.25ex\right\vert\kern-0.25ex\right\vert}}
\newtheorem{definition}{Definition}

\newtheorem{lemma}{Lemma}
\newtheorem{assump}{Assumption}

\newcommand{\myalgoend}[2]{
\noindent\textbf{Output:} #1
\\
\noindent
\rule[0.6\baselineskip]{\textwidth}{0.4pt}
\par}

\newcounter{myalgocounter}[section]
\renewcommand{\themyalgocounter}{
\arabic{myalgocounter}}

\newenvironment{myalgo}[3][]{\refstepcounter{myalgocounter}\par\medskip
  \noindent
  \rule{\textwidth}{0.7pt}
	\vskip -0.1cm
  \noindent 
	\noindent \textbf{Algorithm~\themyalgocounter. #1#2} \rmfamily
	\\
	\noindent
  \rule[0.6\baselineskip]{\textwidth}{0.4pt}
	\vskip -.3cm
	\noindent
  \textbf{Input:} #3 
	}{\medskip 
	}

\newcommand{\condi}{Condition }

\newtheorem{condition}{Condition}

\usepackage{enumitem}

\setlist{nolistsep
}

 \usepackage[square]{natbib}

\newtheorem{example}{Example}

\newcommand{\algname}[1]{Algorithm }

\newcommand{\paper}{paper}
\newcommand{\Xobsi}{X^{i}}
\newcommand{\Xobso}{X^{1}}
\newcommand{\Xobsn}{X^{n}}

\newcommand{\deltagammaone}{\|\hat c_j\|_1}
\newcommand{\deltagammatwo}{\|\hat c_j\|_2}

\newcommand{\betatil}{\ensuremath{\hat {\beta}}}
\newcommand{\gamatil}{\ensuremath{\hat{\gamma}_j}}

\newcommand{\egap}{\ensuremath{\alpha}}

\usepackage{cases,multirow}

\makeatletter
\def\input@path{{/figures/}}
\makeatother

\begin{document}
\title{ \bf
De-biased sparse PCA: 
Inference and testing for eigenstructure of large covariance matrices\\
}
 \author{ 
  \small \bf Jana Jankov\'{a} \and  \small \bf Sara van de Geer 
       }
			\date{\small\it Seminar for Statistics\\ETH Z\"urich
			}
\maketitle
  
\begin{abstract}
Sparse principal component analysis (sPCA) has become one of the most widely used techniques for dimensionality reduction in high-dimen\-sio\-nal datasets.
The main  challenge underlying sPCA is to estimate the first vector of loadings of the population covariance matrix, provided that only  a certain number of loadings are non-zero. 
In this paper, we propose confidence intervals for individual loadings and for the largest eigenvalue of the population covariance matrix. 
Given an independent sample $\Xobsi\in\mathbb R^p,i=1,\dots,n$,  generated from an unknown distribution with an unknown covariance matrix $\Sigma_0$,
we study estimation of the first vector of {loadings} 
in a setting where $p\gg n$. 
Next to the high-dimensio\-nality, another challenge lies in the inherent non-convexity of the problem.
We base our methodology on a Lasso-penalized M-estimator which, despite non-convexity, may be solved by a polynomial-time algorithm such as coordinate or gradient descent. We show that our estimator achieves the minimax optimal rates in $\ell_1$ and $\ell_2$-norm. 
We identify the bias in the Lasso-based estimator and propose a \emph{de-biased sparse PCA} estimator
for the vector of loadings and for the largest eigenvalue of the covariance matrix $\Sigma_0$.
Our main results provide theoretical guarantees for 
asymptotic normality of the de-biased estimator. 
The major conditions we impose are sparsity in the first eigenvector of small order $\sqrt{n}/\log p$ and sparsity of the same order in the columns of the inverse Hessian matrix of the population risk.
\vskip 0.2cm
\noindent

\noindent

\end{abstract}

\section{Introduction}
\subsection{Background and problem}
Principal component analysis (PCA) is a fundamental technique employed for a multitude of tasks including dimension reduction, data visualization and clustering. The applications of PCA range from genomics to image recognition, data
compression and financial econometrics. While in low-dimensional settings, PCA is generally well-understood (see e.g. \cite{andersen}), estimation of eigenstructure in high-dimensional settings has 
opened many intriguing questions. Consequently, the problem 
has attracted substantial interest in the recent decades, see, for example 
\cite{baiksilverstein,paul.pca,jl,amini.wainwright,vu.lei,birnbaum,berthet2013optimal,cai.spca}.

The key challenge underlying the principal component analysis is to estimate the eigenstructure of an unknown population covariance matrix.
In a typical setting, we observe a data matrix $X$ with independent rows $\Xobsi\in\mathbb R^p,i=1,\dots,n$, generated from a $p$-dimensional distribution. Without loss of generality, we assume that $\mathbb E\Xobsi=0$. The population covariance matrix will be denoted by 
$\Sigma_0:=\mathbb E\Xobsi(\Xobsi)^T\in\mathbb R^{p\times p}.$ 
In this \paper, we study estimation and inference for the first \emph{loadings vector} of the population covariance matrix $\Sigma_0$,
defined by
$$\beta_0 :=\text{argmin}_{\beta\in\mathbb R^p} \frac{1}{4}\|\Sigma_0- \beta\beta^T\|_F^2
,$$
where 
$\|\cdot\|_F$
is the Frobenius norm of a matrix. 
The loadings vector $\beta_0$ is an eigenvector of $\Sigma_0$ that satisfies 
$\|\beta_0\|_2^2=\Lambda_{\max}(\Sigma_0)$, where $\Lambda_{\max}(\Sigma)$ denotes the largest eigenvalue of a real symmetric matrix $\Sigma$ and $\|\cdot\|_2$ the Euclidean norm. It defines the best rank-one approximation $\beta_0\beta_0^T$ to the matrix $\Sigma_0$.
We remark that $\beta_0$ is only identifiable up to a sign (meaning that $-\beta_0$ is also a global minimizer),
thus we may choose  this sign  arbitrarily. 
\par
The eigenstructure of the population covariance matrix can be naturally estimated by the eigenstructure of the sample covariance matrix 
$$\hat\Sigma:=\frac{1}{n}\sum_{i=1}^n \Xobsi(\Xobsi)^T.$$
When the dimension $p$ of the observations is fixed, distributional properties of  eigenvalues and eigenvectors
of the sample covariance matrix are well understood:
they are consistent estimators of their population counterparts and have a Gaussian limiting distribution (\cite{andersen}, \cite{var.PC}). If 
$\hat \beta_{\text{PCA}}$ is the first eigenvector of $\hat\Sigma$ rescaled such that 
$\|\hat \beta_{\text{PCA}}\|_2^2 
= {\Lambda_{\max}(\hat\Sigma)}$, 
then under certain regularity conditions
on the eigenvalues of $\Sigma_0$
$$\sqrt{n}(\hat \beta_{\text{PCA}} - \beta_0 )\rightsquigarrow \mathcal N_p(0,V),$$
$$\sqrt{n}(\Lambda_{\max}(\hat\Sigma)-\Lambda_{\max}(\Sigma_0))\rightsquigarrow \mathcal N(0,\sigma^2_{\Lambda}),$$
where $V$ and $\sigma^2_{\Lambda}$ are certain asymptotic variances depending on the distribution of $X.$ 
\par
In a high-dimensional regime, when $p$ is allowed to grow with the sample size, the sample covariance matrix exhibits poor behaviour: \cite{bai.yin} show that the eigenvalues of $\hat\Sigma$ are inconsistent estimators of their population counterparts. Namely
if $p/n\rightarrow \alpha\in(0,\infty),$ then, almost surely
$$\lim_{n\rightarrow \infty} \Lambda_{\max}(\hat\Sigma) = \Lambda_{\max}(\Sigma_0)(1+\sqrt{\alpha})^2.$$ 

\noindent
In the same regime, \cite{johnstone2} shows an analogous statement for the sample eigenvectors. 
In particular, even in a simple model known as the spiked covariance model (studied in numerous works including \cite{jl}, \cite{amini.wainwright}, \cite{montanari.pca}) the sample eigenvectors can be asymptotically perpendicular to the population eigenvectors with high probability. 
More precisely, if
 $\Sigma_0 = I + (\Lambda - 1) uu^T$, with $u^Tu=1,$ $\Lambda\geq 1$ and under technical conditions, as $p/n \rightarrow \alpha > 0$, then almost surely,
\begin{equation}
\label{spca.jl}
\frac{\hat u^T u_1}{\|\hat u\|_2\|u_1\|_2}
\;
 \stackrel{\text{}}{\rightarrow}  \;
\begin{cases}
0 & \text{ if }\Lambda-1\leq \sqrt{\alpha}\\
\frac{1-\alpha/(\Lambda-1)^2}{1+\alpha/(\Lambda-1)^2} & \text{ if }\Lambda -1 > \sqrt{\alpha},
\end{cases}
\end{equation}
where $\hat u$ is the first eigenvector of $\hat\Sigma $ and $\eig$ is the first eigenvector of $\hat u$. 
These results show the inconsistency of $\hat u$, however, note that they  also suggest that consistent estimation  might be possible if $\alpha /(\Lambda-1)^2 \rightarrow 0$, that is if the gap between the largest and second largest eigenvalue of $\Sigma_0$, $\Lambda-1$, grows at least as fast as $\sqrt{p/n}.$  This special but interesting setting  has recently attracted substantial interest and we will remark on it in Section \ref{subsec:relwork} on related literature. 

\par
The above results show that consistent estimation of eigenstructure in high-dimensional settings is not possible without further structural assumptions.
However, in many applications, it is inevitable that the number of variables $p$ is of the same order or even much larger than the sample size $n$.  
This motivated research in sparse settings, where  the first few population eigenvectors  are assumed to only have a certain number of entries non-zero.
Examples of settings where sparse representations are relevant include micro-array studies in genetics or EEG studies of the heart, where the heart-beat cycle  may be expressed in a sparse wavelet basis (see \cite{jl}).
Under sparsity conditions, consistent estimation of the eigenstructure becomes possible. A large body of literature studies methodology and lower bounds for estimation of the population eigenstructure. 
A simple and popular methodology is based on thresholding of the sample covariance matrix, which was investigated mostly  within the spiked covariance model \citep{jl,amini.wainwright,montanari.pca}. 
Methods exploiting Lasso penalization were studied among others in \cite{jolliffe2003modified} and
 \cite{zou.pca}; these however lead to non-convex problems which pose computational difficulties. 
The paper \cite{aspremont.direct} addresses the non-convexity problem by deriving a semidefinite programming-based relaxation for the Lasso-penalized principal component analysis, which was later extended by \cite{fantope}. 
Important work on lower bounds for estimation of eigenstructure includes \cite{vu.lei}, \cite{berthet2013optimal} and \cite{cai.spca}.
In particular, \cite{vu.lei}
propose an estimator $\hat Z$ of $\eig\eig^T$ which achieves the minimax rate, namely, with probability tending to one,
\begin{equation}
\label{minim}
\|\hat Z - \eig \eig^T \|_F^2 \leq \frac{C}{(\Lambda_1-\Lambda_2)^2}s\lambda^2.
\end{equation}
where $s:=\|\eig\|_0$ is the sparsity of the first eigenvector, $\lambda\asymp \sqrt{\log p/n}$ and $C$ is a universal constant. The estimator $\hat Z$ is not computable in polynomial time, however, they propose a polynomial-time estimator which achieves a somewhat slower rate, namely $s^2 \lambda^2.$ To achieve the minimax rate with a polynomial-time algorithm may be impossible, see \cite{berthet2013optimal}.

\par 
The literature on estimation of eigenstructure in high-dimensional settings is vast and provides a wide variety of sparsity-inducing estimators. However, these methods do not lead to methodology for inference such as confidence intervals and tests. To the best of our knowledge, asymptotically normal estimation of eigenstructure has yet not been investigated in sparse high-dimensional regimes. 
We aim to contribute to filling this practical and theoretical gap,
in particular, we address construction of confidence intervals for entries of the first loadings vector $\beta_0$ and the largest eigenvalue of $\Sigma_0$.

\subsection{Outline of methodology, results and contributions }
We briefly summarize the main contributions of this \paper.
We base our construction of asymptotically normal estimators of $\beta_0$ on a Lasso-regularized M-estimation procedure of type
\begin{equation}
\label{spca.intro.beta}
\hat\beta\in \text{argmin}_{\beta\in\mathcal B}\frac{1}{4} \|\hat\Sigma-\beta\beta^T\|_F^2 + \lambda\|\beta\|_1,
\end{equation}
where $\|\cdot\|_1$ is the $\ell_1$-norm  and $\mathcal B$ is a certain local set that guarantees convexity of the population loss function. The local set will be obtained from an initial rough estimator. We will then show in Theorem \ref{oracle} that any stationary point of the program \eqref{spca.intro.beta} is a near-oracle estimator of $\beta_0$ and that it achieves near-oracle rates
in $\ell_2$-norm, namely $\|\hat\beta-\beta_0\|_2^2=\mathcal O_P(s\log p/n)$. 
Since we use localization first, we are able to achieve the minimax rates \eqref{minim} even with a polynomial-time algorithm. 
\par
The estimator  $\hat\beta$ is asymptotically biased; consequently, we identify the bias term and propose methodology to estimate it, which leads to a de-biased estimator. 
Our main theoretical results in Theorem \ref{ci} show that a de-biased sparse PCA estimator leads to asymptotically normal estimators for the entries of the first loadings vector $\beta_0.$ 
We also propose an  estimator for the largest eigenvalue of $\Sigma_0$ and provide theoretical guarantees on the limiting distribution in Theorem \ref{eigenvalue}.
Moreover, the asymptotic variance of the Gaussian limiting distribution corresponds to the asymptotic variance of asymptotically efficient estimation in the low-dimensional setting. An implication of our work is that we require  the sparsity condition is $s=o(\sqrt{n}/\log p)$ in $\beta_0$ and sparsity in the inverse Hessian matrix of the population risk at $\beta_0$.

In an empirical study, we show that our method performs well even when the classical PCA fails, the gain is especially visible in regimes when $p$ is of the same order as $n$ and the eigenvalue gap is relatively small.

\par
\subsection{Related literature}
\label{subsec:relwork}
In this section, we discuss prior related work and outline the differences to our settings and results.
The recent papers  \cite{fan.spca} and the line of papers
\cite{koltchinskii2016asymptotics},\cite{kolt.normal}, \cite{2016arXiv160101457K} and \cite{2017arXiv170807642K},
study asymptotically normal estimation of eigenstructure in high-dimensional settings.
However, their setting and results substantially differs from ours. Their setting essentially requires that the maximum eigenvalue or the eigenvalue gap diverges (see the comment following equation \eqref{spca.jl} above).
Therefore, thanks to this structural assumption, the papers \cite{koltchinskii2016asymptotics} and
\cite{fan.spca} manage to study the high-dimensional setting $p\gg n$ and do {not} require any sparsity conditions.
 We study the setting where the eigenvalue gap may be even very small, thus our situation becomes more difficult, which requires that we impose sparsity conditions. 

We briefly discuss their contributions below. 
The paper \cite{koltchinskii2016asymptotics} derives the asymptotic distribution of the leading sample eigenvector of the sample covariance matrix in a setting where $p$ is allowed to grow with the sample size. This is established under the ``effective rank'' condition 
$\Lambda_{\max}(\Sigma_0) \gg \text{tr}(\Sigma_0) /n$, where $\Lambda_{\max}(\Sigma_0)$ is the maximum eigenvalue of $\Sigma_0$ and under a Gaussianity assumption.
They show that in this setting the leading eigenvector is biased. The paper then proposes a way of estimating this bias via sample splitting and constructs a de-biased  estimator which is asymptotically normal. Interestingly, their results imply that in special high-dimensional settings where the effective rank condition holds, consistent estimation is possible even if $p/n\rightarrow \infty$, without imposing sparsity assumptions.

\par
The paper \cite{fan.spca} (see also a related paper \cite{shen2013surprising}) provides similar results as \cite{koltchinskii2016asymptotics}, but considers the spiked covariance model. In particular, it is required that the first $d$ eigenvalues of $\Sigma_0$ diverge to infinity (denoting the eigenvalues by $\Lambda_j,j=1,\dots,d$, they must satisfy the condition  $\Lambda_{j}\geq \sqrt{p/n}$) and the non-spiked eigenvalues are assumed to be bounded.
This means that the eigenvalue gap (the difference between the smallest eigenvalue in the spiked part and 
the largest eigenvalue in the non-spiked part) must grow at least at the rate $\sqrt{p/n}$. 
 Under this condition and a sub-Gaussianity condition, they derive the asymptotic distribution of the first $d$ eigenvalues and the corresponding eigenvectors of the sample covariance matrix.
Similarly as in \cite{koltchinskii2016asymptotics}, their results reveal a bias in the asymptotic distribution,
in particular the spiked eigenvalues $\Lambda_{j}(\hat\Sigma) $ of $\hat\Sigma$ satisfy, for $j=1,\dots,d$, 
\begin{equation}
\label{fan.spca}
\sqrt{n}\left[\frac{\Lambda_{j}(\hat\Sigma) }{\Lambda_{j}} -1-
 \left(\frac{C p}{n\Lambda_{j}} + \mathcal O_P\left(\Lambda_j^{-1}\sqrt{\frac{p}{n}}\right)\right)
\right] \rightsquigarrow \mathcal N(0, \kappa_j-1),
\end{equation}
where 
$\kappa_j$ is a certain measure of kurtosis. The asymptotic bias term is $\frac{C p}{n\Lambda_{j}}$, where  the constant $C$ is unknown.
The authors  propose a shrinkage estimator based on soft-thresholding which also involves  a bias correction based on equation \eqref{fan.spca}, and the unknown $C$ is replaced by a consistent estimator.
\par

\subsection{Organization of the \paper}

We discuss the properties and non-convexity of the population risk function in Section \ref{sec:spca.prelim}
a propose a first-step estimator which is guaranteed to reach a local neighbourhood of the true underlying parameter where the population risk function is convex. In Section \ref{sec:spca.main}, we provide the main methodology for an oracle estimator of the first loadings vector and asymptotically normal estimators of the loadings vector and the maximum eigenvalue of $\Sigma_0$, and establish our main theoretical results. 
In Section \ref{sec:sim}, we investigate the performance of our methodology in an empirical study. 
Section \ref{sec:disc} discusses conclusions and implications
and the proofs are deferred to Section \ref{sec:proofs}.

\subsection{Notation} 
For a vector $x\in\mathbb R^d$, we let $x_j$ denote its $j$-th entry. For a matrix $A\in\mathbb R^{m\times d}$ we use the notation $A_{ij}$ or $(A)_{ij}$ for its $(i,j)$-th entry and $A_j$ to denote its $j$-th column. 
We let 
$\vertiii{A}_\infty=\max_{i} \|e_i^T A\|_1$, where $e_i$ is the $i$-th unit vector, $\vertiii{A}_1= \vertiii
{A^T}_\infty$, $\|A\|_\infty=\max_{i,j}|A_{ij}|$ and the Frobenius norm is denoted by $\|A\|_F=(\sum_{i,j} A_{ij}^2)^{1/2}$.
By $\Lambda_{\min}( A)$ and $\Lambda_{\max}(A)$ we denote the minimum and maximum eigenvalue of $A$, respectively. 
For sequences of random variables $X_n,Y_n$, we write $X_n=\mathcal O_P(Y_n)$ if $X_n/Y_n$ is bounded in probability. We write $X_n=o_P(1)$ if $X_n$ converges to zero in probability and we use $\rightsquigarrow$ to denote convergence in distribution.

\section{Preliminaries}

\label{sec:spca.prelim}

\subsection{Landscape of population risk and non-convexity}

In this section, we introduce the setup and develop methodology to obtain an initial estimator of the vector of loadings $\beta_0$ in a high-dimensional setting, under a sparsity assumption on the entries of $\beta_0$. The main methodology for construction of an
asymptotically normal estimator is given in Section \ref{sec:spca.main}.
\par
\noindent
The spectral decomposition of $\Sigma_0$ is given by
$$\Sigma_0= U^T \Phi^2 U,$$
where $\Phi:=\text{diag}(\phi_1,\dots,\phi_p)$ and we assume that 
$$\phi_1>\phi_2 \;\;\;\text{ and }\;\;\;\phi_2\geq  \dots\geq \phi_p\geq 0,$$ 
and $U=(u_1,\dots,u_p)$ is such that  $UU^T =I$.
Note that while $\Phi$ is unique, the matrix $U$ is in general not unique if the eigenvalues have multiplicities.
We do not require that $U$ is unique, however, we require that the first eigenvector $u_1$ is unique (up to a sign): 
this is the case if $\phi_1 > \phi_2.$ 
The eigenvalues of $\Sigma_0$ will be denoted by $\Lambda_j := \phi_j^2 $ for $j=1,\dots,p$. We also use the alternative notations  $\phi_{\max}:=\phi_1 $ and $\Lambda_{\max}:=\Lambda_1.$
The gap between the square-root of the largest and second largest eigenvalue of $\Sigma_0$ will be denoted by 
$$\rho: = \phi_{1}- \phi_2.$$
Note that this definition also implies that $\Lambda_{1} -\Lambda_2 = (\phi_1-\phi_2)(\phi_1+\phi_2)= \rho^2 + 2\rho\phi_2 \geq \rho^2.$ 
We will refer to both $\rho$ and 
$\Lambda_{1}-\Lambda_2$ as the ``eigenvalue gap'', depending on the context.
The  eigenvalue gap determines the curvature of the population risk and thus naturally plays an intrinsic role in estimation of the related eigenspaces: if the eigenvalue gap 
vanishes too fast, consistent estimation of the first eigenvector becomes impossible.
\\
\noindent
Our main methodologies are based on the (regularized) M-estimation fra\-me\-work. To this end, we consider the theoretical risk function
$$R(\beta):=\frac{1}{4} \|\Sigma_0-\beta\beta^T\|_F^2 = \frac{1}{4}\text{tr}(\Sigma_0)-\frac{1}{2}\beta^T \Sigma_0 \beta + \frac{1}{4}\|\beta\|_2^4.$$
The risk function is plotted in Figure \ref{fig:spca} for the simple case $p=2$.

\begin{figure}[h!]
\centering 
\includegraphics[width=0.46\textwidth]{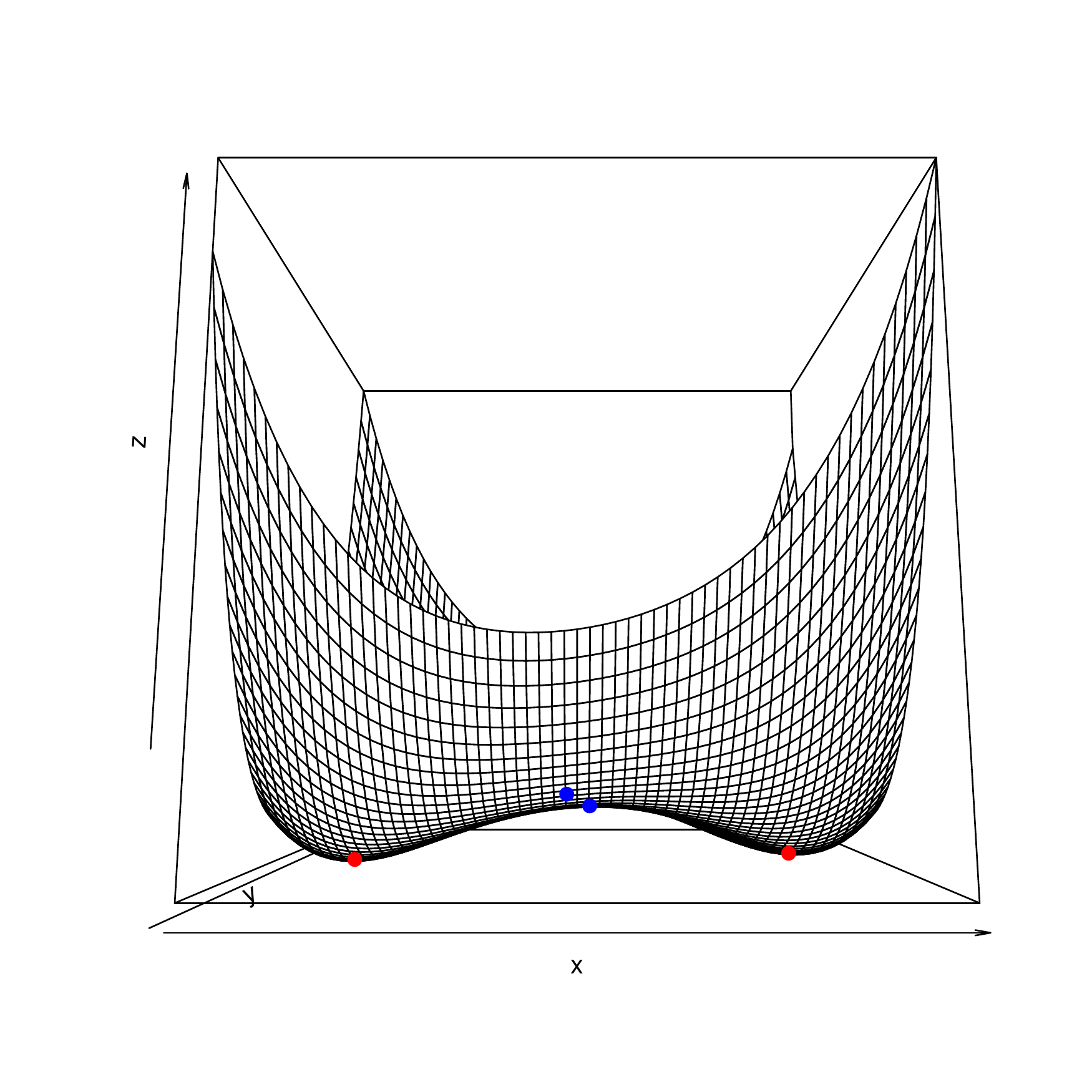}
\includegraphics[width=0.46\textwidth]{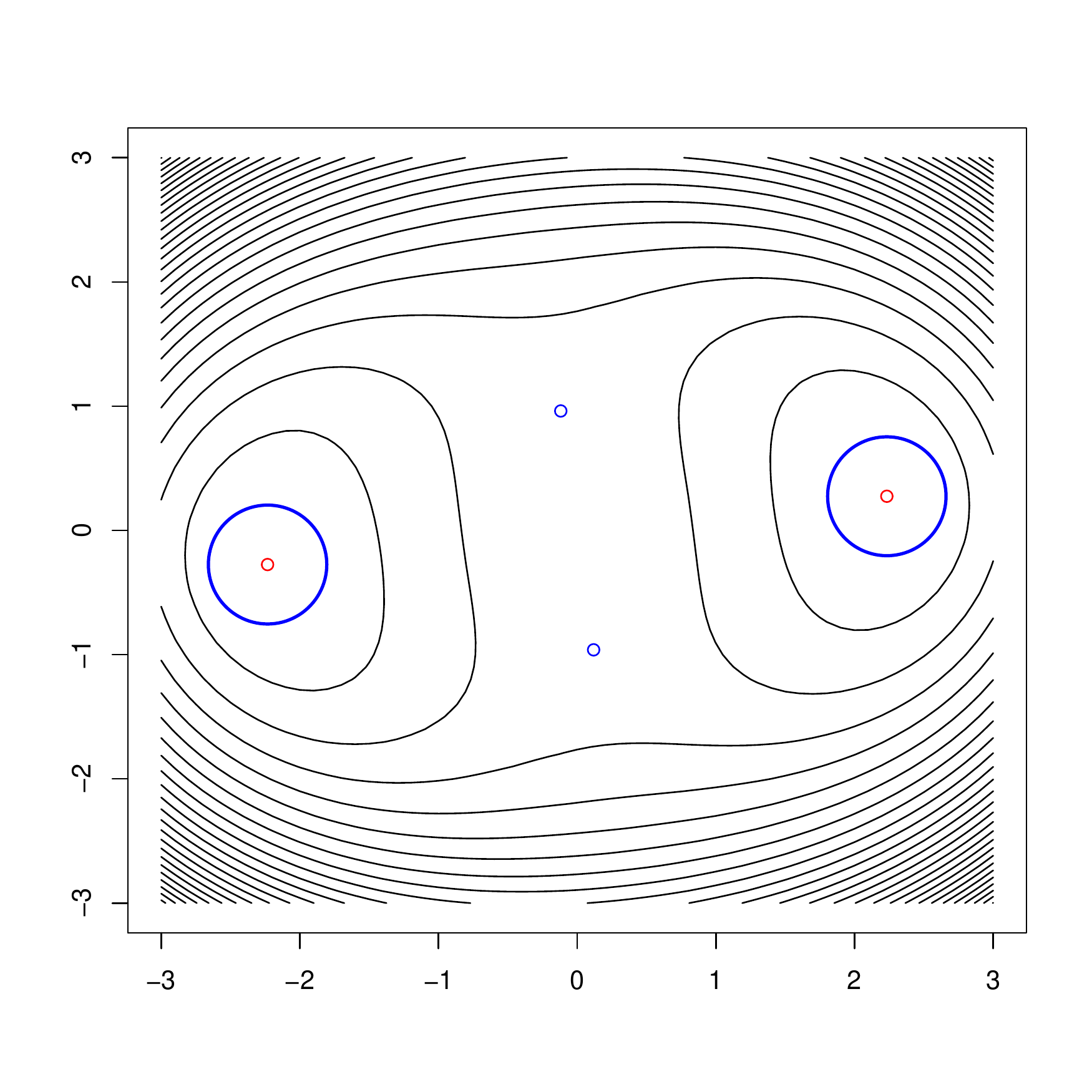}
\caption[Population risk function for the principal component analysis]{A toy example: the graph (left) and contours (right) of the \textbf{population risk} for $p=2$. The two global minima are labeled by red points, the two saddle points by blue points.
The blue circles in the contour plot show the convex local neighbourhood from Lemma \ref{eigen.spca} below.}
\label{fig:spca}
\end{figure} 

\noindent
The gradient $\dot R(\beta)$ and the Hessian $\ddot R(\beta)$ of $R(\beta)$ are given by
$$\dot R(\beta) = -\Sigma_0 \beta +  \|\beta\|_2^2\beta,$$
$$\ddot R(\beta) = -\Sigma_0  +  \|\beta\|_2^2I +2 \beta\beta^T.$$
\noindent 
We consider the empirical analogue of $R(\beta)$, the empirical risk function
$$R_n(\beta) = \frac{1}{4}\|\hat\Sigma -\beta\beta^T\|_F^2 .
$$
The gradient and Hessian of $R_n$ will be denoted by $\dot R_n(\beta)$ and $\ddot R_n(\beta)$, respectively.
This choice of a risk function allows us to formulate estimation of $\beta_0$ in the M-estimation framework. 
However,  a simple naive approach via minimizing the empirical risk  $R_n(\beta)$ is plagued by non-convexity: even the population risk $R(\beta)$ itself is a non-convex function on $\mathbb R^p$. If $\phi_1>\phi_2$, 
the population risk has a unique (up to sign) global minimizer $\phi_1u_1\equiv \beta_0$, however, it is well known that computing the global minimizer of a non-convex function is a difficult problem. 
It is easy to deduce that the population risk has stationary points which are given by
$\pm\phi_j u_j, j=1,\dots,p,$ where $u_j$ is \emph{any} normalized eigenvector corresponding to $\phi_j$ or $u_j$ is the zero vector. Thus, the population risk might have a continuum of stationary points (consider e.g. the degenerate case $\Sigma_0=I$: the stationary points form a sphere if we disregard the zero vector).
In the simple case when there are no eigenvalue multiplicities, there are $2p+1$ stationary points: the points $\pm\phi_1u_1$ are the global minimizers and one can easily deduce that the remaining stationary points (except zero) are all saddle points by inspecting the Hessian matrix 
$$\ddot R(\phi_ju_j) = \sum_{i=1}^p (\Lambda_j - \Lambda_i)u_iu_i^T + 2\Lambda_j u_ju_j^T
.$$
The strategy we will employ to overcome the non-convexity of the population risk is based on the observation that locally around the true $\beta_0$, the population risk function $R(\beta)$ is convex, as illustrated in Figure \ref{fig:spca}. 
Lemma \ref{eigen.spca} below relates the eigenvalue gap $\rho$ to the convexity of the population risk function: in
an $\ell_2$-ball around $\beta_0$ whose radius is small enough compared to the eigenvalue gap, the population risk is convex.

\begin{lemma}[Lemma 12.7 in \cite{sf}]\label{eigen.spca}
Suppose that $3\eta <\rho.$
Then for all $\beta\in \mathbb R^p$ satisfying $\|\beta - \beta_0\|_2\leq \eta$ we have 
$$\Lambda_{\min}(\ddot R(\beta)) \geq 2(\rho-3\eta).$$
\end{lemma}

\par 
 However, note that the statement of Lemma \ref{eigen.spca} is not  necessarily true for the empirical risk $R_n(\beta)$. In high-dimensional settings, the empirical risk might be non-convex even in the local neighbourhood from Lemma \ref{eigen.spca}, because it depends on the sample covariance matrix $\hat\Sigma$ whose eigenvalues are inconsistent estimators of the population eigenvalues and might even diverge to infinity in the regime $p\gg n $; for illustration of the empirical risk function, see Figure \ref{fig:spca2}. 

Following the idea of Lemma \ref{eigen.spca}, our strategy is to estimate the loadings vector $\beta_0$ using a two step procedure. In the first step, we localize to an $\ell_2$-ball around $\beta_0$, which
is small enough such that $\rho-3\eta>0.$ In the second step, we make use of the locality to obtain a near-oracle estimator.

\subsection{Localization: first step estimator}

\noindent

\begin{figure}[h!]
\centering 
\includegraphics[width=0.48\textwidth]{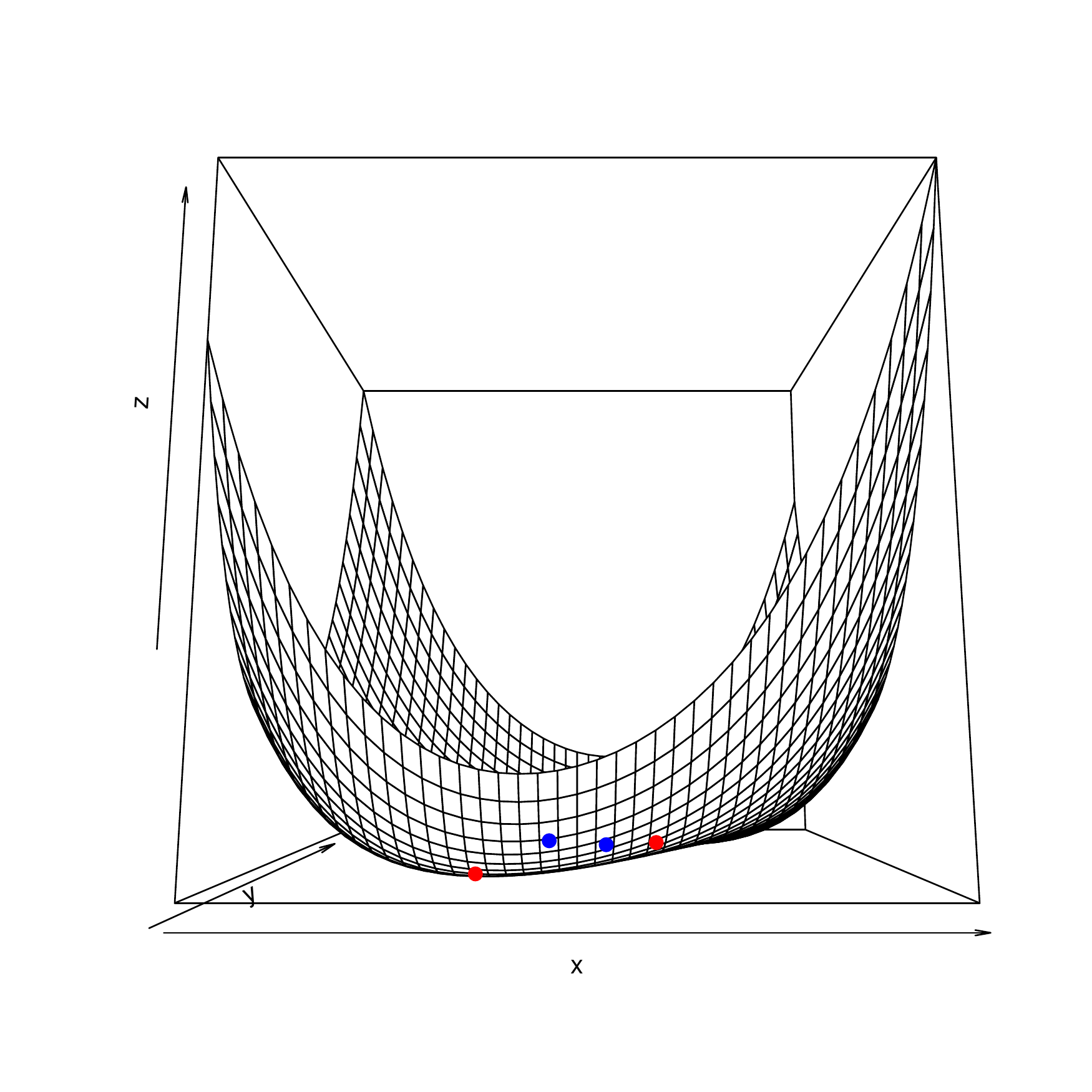}
\includegraphics[width=0.46\textwidth]{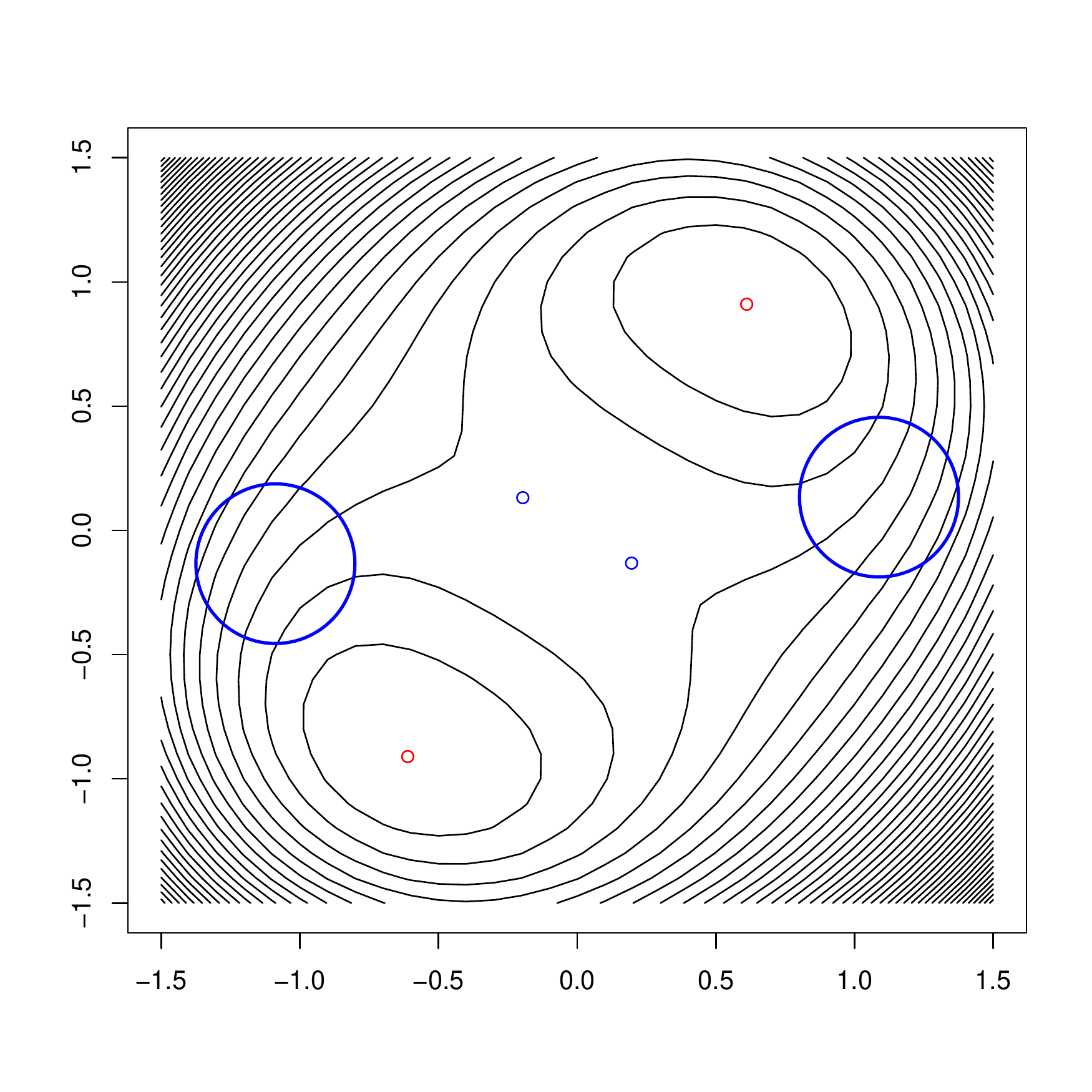}
\caption[Empirical risk function for the principal component analysis]{A toy example: the graph (left) and contours (right) of the \textbf{empirical risk} for $p=2$ and a randomly generated sample of size $ n=4$ (the observations are normally distributed in this example). 
The blue circles in the contour plot show the convex local neighbourhood from Lemma \ref{eigen.spca} (around the true $\beta_0$).}
\label{fig:spca2}
\end{figure}

\par
We base our first-step estimator on a convex program originally proposed in \cite{aspremont.direct} (and later studied by \cite{fantope}),
\vskip 0.1cm
\begin{equation}
\label{hatZ}
\hat Z:= \operatornamewithlimits{argmax}_{\text{tr}(Z) = 1,\; 0 \;\preceq\; Z \;\preceq \;I} \;\text{tr}(\hat\Sigma Z)-\lambda\|Z\|_1.
\end{equation}
\vskip 0.1cm
The feasible set is a convex relaxation of the set of positive definite rank-one matrices and the $\ell_1$-norm 
of a matrix is the $\ell_1$-norm of its vectorized version. Note that due to the relaxation, $\hat Z$ is not necessarily of rank one.
However, we show that the normalized eigenvector of $\hat Z$ corresponding to its largest eigenvalue, denote it by $\hat u_1$, may be used to estimate the corresponding population eigenvector $u_1$ (up to a sign).
We then define an initial estimator of $\beta_0$ as a properly scaled version of $\hat u_1,$
\begin{equation}
\label{beta.init.spca}
\hat \beta_{\text{init}}:= \text{tr}(\hat\Sigma \hat Z)^{1/2} \hat u_1.
\end{equation}
Lemma \ref{spca.init} below provides guarantees for the estimator $\hat\beta_{\text{init}}$ under mild conditions.
To this end, we recall Theorem 3.3  in \cite{fantope} which derives the bound for $\hat Z$ in Frobenius norm. 
By the standard arguments for deriving oracle inequalities for $\ell_1$-penalized M-estimators (see e.g. \cite{hds})
the bound from Theorem 3.3 in \cite{fantope} can be easily extended to the $\ell_1$-norm error. Recall that $u_1$ is the eigenvector of $\Sigma_0$ corresponding to its largest eigenvalue. Then
 for $\lambda\geq 2\|\hat\Sigma-\Sigma_0\|_\infty$,  it holds
\begin{equation}
\label{fantope.bound}
\|\hat Z - \eig \eig^T\|_F^2 + \lambda\|\hat Z - \eig\eig ^T\|_1\leq \frac{Cs^2\lambda^2}{(\Lambda_{1}-\Lambda_2)^2} =:\epsilon^2
,
\end{equation}
where $s$ is the number of non-zero entries in $\beta_0$ and $C>0$ is a universal constant.
\begin{lemma}
\label{spca.init}
Let $\hat Z$ be the estimator defined in \eqref{hatZ},
with $\lambda \geq 2\|\hat\Sigma-\Sigma_0\|_\infty$.
Letting $\hat u_1$ denote the normalized eigenvector of $\hat Z$ corresponding to its largest eigenvalue and assuming $\hat u_1^T \eig\geq 0$, it holds
$$\|\hat u_1 - \eig\|_2^2\leq 4\epsilon,$$
where $\epsilon$ is defined in \eqref{fantope.bound}, 
and
$$\|\hat \beta_{\emph{init}} - \beta_0\|_2 \leq  \frac{\zeta}{4\sqrt{\|\beta_0\|_2^2-\zeta}} + 2\|\beta_0\|_2\sqrt{\epsilon}
  ,$$
where 
$$\zeta:={s\lambda}+ {\epsilon^2}  + 6\|\beta_0\|_2^2 \epsilon  +  4\|\beta_0\|_2\sqrt{\epsilon},
$$
provided that we assume   $\|\beta_0\|_2^2-\zeta >0.$
\end{lemma}
Under a sub-Gaussianity condition on the design (as will be assumed below in Section \ref{subsec:spca.theory}), 
Lemma \ref{spca.init} implies that with $\lambda\asymp \sqrt{\log p/n}$, it holds that
$$\|\hat \beta_{\text{init}} - \beta_0\|_2=\mathcal O_P(\sqrt{s\lambda}),$$
 provided that
 $\|\beta_0\|_2=\mathcal O(1), 1/(\Lambda_1 -\Lambda_2)=\mathcal O(1)$ and $s\lambda\rightarrow 0$.


\section{Main results }
\label{sec:spca.main}
\subsection{Methodology and de-biasing}

\label{sec:spca.main.method}
In this section, we  define the second step estimator  and  propose metho\-dology for asymptotically normal estimation
 of loadings and the maximum eigenvalue of $\Sigma_0$. 

We aim to define the second-step estimator localized in an $\ell_2$-neighbour\-hood of the initial estimator $\hat\beta_{\text{init}}$.
However, for simplicity of presentation, we will define the local  neighbourhood around $\beta_0$ instead of $\hat\beta_{\text{init}}$ as follows, 
$$\mathcal B_{}:=\{\beta\in\mathbb R^p: 
\|\beta-\beta_0
\|_2\leq \eta\},$$
where $\eta$ is some suitable positive constant. 
In practice we replace $\beta_0$ in the above definition by $\hat\beta_{\text{init}}$;
then Lemma \ref{spca.init} provides guarantees that for $n$ sufficiently large, a small $\ell_2$-neighbourhood around $\hat\beta_{\text{init}}$
will contain  $\beta_0$ with high probability.
We define the program
\begin{eqnarray}
\label{est.fast}
\hat \beta \in
\operatornamewithlimits{argmin}_{\beta \in \mathcal B,\; \|\beta\|_1\leq \tune} R_n(\beta)
+ \lambda \|\beta\|_1,
\end{eqnarray}
where $(\lambda,\tune)$ is a pair of positive tuning parameters.
We include the constraint $\|\beta\|_1\leq \tune$ due to non-convexity of $R_n(\beta)$. This will be necessary for deriving theoretical guarantees for $\hat\beta$, namely for bounding the probabilistic error term. This constraint is not restrictive, but requires to provide a value for the tuning parameter $\tune$. Asymptotically, $T\asymp 1/\lambda \asymp \sqrt{n/\log p}.$  Similar constraints were studied e.g. in \cite{non-convex1}.
\par
As pointed out previously, the optimized function \eqref{est.fast} may be non-convex even over the local set $\mathcal B$. 
Hence it may possess stationary points that are not global optima. Iterative methods such as gradient or coordinate descent 
are guaranteed to eventually converge to a stationary point, regardless of convexity, but this point could be a local minimum, saddle point or even a local maximum. 
Otherwise computing global optima of non-convex functions in an efficient manner may be very difficult in practice.
To overcome this difficulty, we provide statistical guarantees for any stationary point of the program \eqref{est.fast}, not only for the global minimizer. 
Similar statistical guarantees providing oracle inequalities for non-convex regularized M-estimators were studied e.g.  in
\cite{non-convex1} or \cite{sf}. 
A stationary point $\betatil$ of the program \eqref{est.fast} is any  point of the feasible set where
\begin{equation}
\label{stationary}
(\dot R_n(\betatil) +\lambda\partial\|\betatil\|_1)^T(\beta - \betatil) \geq 0, \quad\text{ for all }\beta\in\mathcal B, \|\beta\|_1\leq \tune,
\end{equation}
where $\partial\|\betatil\|_1$ denotes the sub-differential of $\|\beta\|_1$ evaluated at $\betatil.$
This definition accounts  also for local minima at the boundary; if the stationary point lies in the interior of the feasible set, then \eqref{stationary}
reduces to the Karush-Kuhn-Tucker (KKT) conditions $\dot R_n(\betatil) +\lambda\partial\|\betatil\|_1=0.$

In the next section, we show that any stationary point $\hat\beta$ is a near-oracle estimator of $\beta_0.$
However, it is asymptotically biased as will be shown in the sequel, but we can employ de-biasing (or de-sparsifying) techniques studied in \cite{vdgeer13}. 
If $\hat\beta$ is a stationary point defined as in \eqref{stationary}, the de-sparsifying  approach suggests to take
the bias-corrected ``estimator''
\begin{equation*}
\tilde b:= \hat\beta - \Theta_0 (\|\hat\beta\|_2^2 \hat\beta-\hat\Sigma \hat \beta  ),
\end{equation*}
where $\Theta_0$ is 
 the inverse Hessian matrix of the population risk, $\Theta_0 := \ddot R(\beta_0)^{-1}.$
The $p\times p$ matrix $\Theta_0$ is not known and needs to be replaced by a consistent estimator as will be proposed below.
\par
Furthermore, we aim to construct an asymptotically normal estimator for the maximum eigenvalue, which is a quadratic function of $\beta_0.$
This estimation problem was not considered in \cite{vdgeer13}, but similar ideas may be applied. We will show that the estimator $\|\hat\beta\|_2^2$ is biased for $\Lambda_{\max}$, but may be de-biased by defining 
\begin{equation*}
\tilde\Lambda:= \|\hat\beta\|_2^2 - 2\hat\beta^T\Theta_0^T(\|\hat\beta\|_2^2 \hat\beta-\hat\Sigma \hat \beta ).
\end{equation*}
An estimator of $\Theta_0$ may be constructed in a similar spirit as in \cite{vdgeer13} using nodewise regression. 
Nodewise regression was  studied in \cite{vdgeer13} for generalized linear models which have a special structure in the Hessian matrix of the empirical risk and the empirical Hessian matrix is positive semi-definite. We however aim to apply nodewise regression to approximately invert the Hessian matrix
$$\ddot R_n(\hat\beta):=-\hat\Sigma + \|\hat\beta\|_2^2 I + 2\hat\beta\hat\beta^T,$$
where the special structure from generalized semi-linear models is not present and moreover, the empirical Hessian is not necessarily positive definite.
To deal with the non-convexity which arises due to absence of positive semi-definiteness, we modify the nodewise regression program from \cite{vdgeer13} by adding an extra constraint $\|\cdot\|_1\leq T$ with a tuning parameter $T>0.$ Moreover, 
due to non-convexity, we need to derive oracle inequalities for any stationary point instead of only the global minimum.
\par
In Algorithm \ref{alg:nodewisea} below we formulate the modified version of the nodewise regression program for an arbitrary input matrix $A$. 
Recall that for a matrix $A\in \mathbb R^{p\times p}$, we let $A_j$ denote its $j$-th column, $A_{-j}\in\mathbb R^{p\times (p-1)}$ the matrix $A$ without its $j$-th column, $A_{j,-j}\in\mathbb R^{(p-1)\times 1}$ denote 
the column vector obtained by selecting the $j$-th row of $A$ and removing its $j$-th entry, 
and by $A_{-j,-j}\in \mathbb R^{(p-1)\times (p-1)}$ we denote the matrix 
$A$ without its $j$-th column and the $j$-th row.
\vskip -0.1cm
\begin{myalgo}{Non-convex Nodewise Lasso}{$A\in\mathbb R^{p\times p}$, positive tuning parameters $(\lambda_j,\tunej),j=1,\dots,p$\vskip 0.1cm}
\label{alg:nodewisea}
\noindent
\textbf{for} $j=1,\dots,p$:\vskip 0.1cm
\begin{itemize}[leftmargin=1cm]
\item[1:]
Compute any stationary point $\hat\gamma_j$ of the program
\begin{equation}
\label{lasso-type}
\operatornamewithlimits{min}\limits_{\gamma_j\in\mathbb R^{p-1}:\;\|\gamma_j\|_1\leq \tunej} 
\Gamma_j^T A_{}\Gamma_j+
\lambda_j \|\gamma_j\|_1,
\end{equation}
where 
\begin{equation}
\label{defi.gama.pca}
\Gamma_j:=(-\gamma_{j,1},\dots,-\gamma_{j,j-1},1,-\gamma_{j,j+1},\dots,-\gamma_{j,p}).
\end{equation}
\item[2:]
Define $\hat\Gamma_j$ via the relation \eqref{defi.gama.pca} with $\hat\gamma_j$ and compute the estimator of the noise level 
$\hat\tau_j^2:= \hat\Gamma_j A_{}\hat\Gamma_{j} + \frac{1}{2}\lambda_j\|\hat\gamma_j\|_1$.
\item[3:]
Compute the nodewise Lasso estimator defined by
$\hat\Theta_{j} := \hat\Gamma_j/\hat\tau_j^2.$
\end{itemize}
\end{myalgo}
Stack $\hat\Theta_j,j=1,\dots,p$ into the columns of $\hat\Theta := [\hat \Theta_1,\dots,\hat\Theta_p]$.\\\noindent
\myalgoend{$\hat\Theta$
}

\noindent
We remark that a stationary point $\hat\gamma_j$ is defined analogously as in \eqref{stationary}, that is, $\hat\gamma_j$ is a stationary point of the program \eqref{lasso-type} if it lies in the feasible set and 
for all $\gamma_j\in\mathbb R^{p-1}$ in the feasible set 
it holds
$$(-2A_{j,-j} + 2A_{-j,-j}\hat\gamma_j +\lambda_j\partial\|\hat\gamma_j\|_1)^T(\gamma_j - \hat\gamma_j) \geq 0,$$
where
 $\partial\|\hat\gamma_j\|_1$ is the sub-differential of the $\ell_1$-norm evaluated at $\hat\gamma_j$.
If $\hat\gamma_j$ is a stationary point of the program \eqref{lasso-type} which lies in the interior of the feasible set, 
then 
\begin{equation}
\label{kkt.spca.gamma}
-2A_{j,-j} + 2A_{-j,-j}\hat\gamma_j +\lambda_j\partial\|\hat\gamma_j\|_1=0.
\end{equation}
In this case, using the KKT conditions \eqref{kkt.spca.gamma}, one can show  
 that
$$A_j^T\hat\Gamma_j =\hat\tau_j^2\quad\text{ and }\quad\|A_{-j}^T\hat\Gamma_j\|_\infty \leq\lambda_j/2,$$
which implies
 $$\|A^T \hat\Theta - I \|_\infty =\mathcal O( \max_{j=1,\dots,p}\lambda_j /\hat\tau_j^2).$$
We aim to apply the nodewise Lasso with $A:=\ddot R_n(\hat\beta)$ and for this choice, we show in the following section that we can obtain an oracle inequality for $\hat\Theta_j$.
 Our theoretical results also identify the (asymptotically) correct choice of the tuning parameters $\lambda\asymp 1/\tune \asymp \sqrt{\log p/n}$. From a computational viewpoint, calculating any stationary point of the Lasso-type program \eqref{lasso-type} can be achieved by a polynomial time algorithm (such as the gradient or coordinate descent). 

\noindent
We collect the full procedure for obtaining the de-biased estimator in the scheme below. 

\begin{myalgo}{{De-biased sparse PCA}}{$n\times p $ data matrix $X$, positive tuning parameters $\lambda_{\text{init}},$ $ (\lambda,\tune), (\lambda_j, \tunej),$ \\$j=1,\dots,p$}
\label{des-PCA}
\end{myalgo}
\begin{itemize}[leftmargin=*]
\item[1:]
Compute the initial estimator
$\hat \beta_{\text{init}} 
$ defined in \eqref{beta.init.spca} with the tuning parameter $\lambda_{\text{init}}$
\item[2:]
Compute any stationary point $\hat \beta$ of the following program, with tuning parameters $(\lambda,\tune)$
\begin{equation}
\label{spca.posledna}
\operatornamewithlimits{argmin}\limits_{\beta \;\in\;\mathbb R^p:\;\|\beta\|_1\leq T, \;\;\|\beta-\hat\beta_{\text{init}}\|_2\leq \eta} -\frac{1}{2}\beta^T \hat\Sigma \beta + \|\beta\|_2^4+ \lambda \|\beta\|_1.
\end{equation}
\item[3:]
Run the nodewise Lasso in Algorithm \ref{alg:nodewisea} with input matrix 
$$\ddot R_n(\hat\beta)= -\hat\Sigma + \|\hat\beta\|_2^2 I + 2\hat\beta\hat\beta^T,$$
with tuning parameters $(\lambda_j,\tunej),$ $j=1,\dots,p$ and output $\hat\Theta$
\item[4:]
 Compute the de-sparsified estimator and the eigenvalue estimator:
\begin{equation}
\label{despca.b}
\hat b :=  \hat\beta - \hat\Theta^T (\|\hat\beta\|_2^2 \hat\beta-\hat\Sigma \hat \beta  ),
\end{equation}
\begin{equation}
\label{despca.lambda}
\hat \Lambda := \|\hat \beta\|_2^2 - 2\hat\beta^T\hat\Theta^T(\|\hat\beta\|_2^2 \hat\beta-\hat\Sigma \hat \beta ).
\end{equation}
\end{itemize}

\myalgoend{$\hat b,\hat\Lambda$}

The tuning parameters in Algorithm \ref{des-PCA} have to be chosen of order $\lambda_{\text{init}}\asymp \lambda\asymp 1/\tune\asymp \lambda_j\asymp 1/\tunej\asymp\sqrt{\log p/n}$ and the constant $\eta$ in program \eqref{spca.posledna}  must be chosen sufficiently small.

\noindent
\subsection{Theoretical results
}

\label{subsec:spca.theory}
In this section, we derive the main theoretical results: 
firstly we provide oracle inequalities for $\hat\beta$ and the nodewise Lasso $\hat\Theta$ (thoughout this section, $\hat\beta$ is the estimator defined in \eqref{est.fast}, where we assume we are already in the neighbourhood around $\beta_0$ and $\hat\Theta$ is based on $\hat\beta$); 
secondly we provide results on asymptotic normality of the bias-corrected estimators based on $\hat\beta$ and $\hat\Theta$. 
We discuss how these results may be used to construct confidence intervals and support recovery.

 To bound the terms arising from the probabilistic analysis of the estimators, we assume sub-Gaussian design, but we remark that
similar results could be obtained under bounded design using the concentration results derived in \cite{uniform}.

\begin{definition}
We say that a vector $Y\in\mathbb R^p$ is sub-Gaussian with a parameter $\subgp$ if 
for all vectors $\alpha\in\R^p$ such that $\|\alpha\|_2=1$,  it holds
$$\mathbb E e^{|\alpha^TY|^2/\subgp^2}\leq 2.$$
\end{definition}

\begin{condition}[Sub-Gaussian design]
\label{design.spca}
Assume that the $n\times p$ random matrix $X$ has independent rows, which are sampled from a zero-mean distribution with a covariance matrix $\Sigma_0$ and are sub-Gaussian vectors with a parameter $\sigma.$ We say that $X$ is a sub-Gaussian matrix with a parameter $\sigma.$
\end{condition}

The following lemma derives an oracle inequality for the second step estimator $\hat\beta_{}$.
Recall that $\eta$ is the size of the neighbourhood in the definition of $\hat\beta$, $\rho=\phi_{\max}-\phi_2 $ is the eigenvalue gap
and the sparsity in $\beta_0$ is denoted by
$$s:=\|\beta_0\|_0.$$

\begin{theorem}
\label{oracle}
Assume that Condition \ref{design.spca} is satisfied with a parameter $\sigma$, let $\lambda_0=\sqrt{2\log (2p)/n},$ $
\lambda_1 = 4\sigma^2(\|\beta_0\|_2+1)[\lambda_0+\lambda_0^2]$ and
$$ \rho-3\eta \geq c_0 \sigma^2 C_{\tune} [3C_{\tune}+\sqrt{6}],$$
where $c_0$ is a suitable universal constant.
Let the tuning parameters $(\lambda,\tune)$ of the program \eqref{est.fast} satisfy
\begin{align}
\label{tuning.par1234} 
\lambda\geq 2\lambda_1,
\end{align} 
$\tune \leq C_{\tune} /(2\lambda_0),$
and $\|\beta_0\|_1\leq \tune.$
Then any stationary point $\betatil$ as defined in \eqref{stationary} satisfies with probability at least 
$1-2(J+2)e^{-\log (2p)}$ where $J=\lceil \log T\rceil$ 
 the error bound  
\begin{align}
\label{EqnStatBounds}
\|\betatil - \beta_0\|_2^2 + \lambda \|\betatil - \beta_0\|_1 \le \frac{C_2 s\lambda^2  }{(\rho-3\eta)^2},
\end{align}
where $C_2$ is a universal constant.
\end{theorem}
In an asymptotic formulation, we require that $\|\beta_0\|_1 = \mathcal O(\sqrt{n}/\log p)$ and $\|\beta_0\|_2=\phi_{\max}=\mathcal O(1)$. Then for $\lambda\asymp \sqrt{\log p/n}$ and
$T\asymp \sqrt{n/\log p},$ we obtain rates of order $s\log p/n$, provided that $\rho-3\eta$ is lower bounded by a universal constant. The tuning parameter $\tune$ must be chosen large enough to guarantee that $\beta_0$ lies in the feasible set.
The above result essentially requires that  $\lambda_0 \|\beta_0\|_1$ is bounded by a universal constant for $\hat\beta$
to achieve the oracle rates $s\log p/n.$

Similar oracle inequalities may be derived for the nodewise Lasso estimators; but due to high-dimensionality, sparsity conditions on the columns of $\Theta_0$ are necessary. Sparsity conditions on the inverse population Hessian have appeared in literature on linear regression (\cite{zhang,vdgeer13,stanford1}) and 
generalized linear models (\cite{vdgeer13,vch,vch1}).
For $j=1,\dots,p$, we define the population parameters 
\begin{equation}\label{gamma0}
\gamma^0_{j}:=\operatornamewithlimits{argmin}_{\gamma\in\mathbb R^{p-1}} \;\;
\Gamma_j^T \ddot R(\beta_0)\Gamma_j,
\end{equation}
where $\Gamma_j$ is defined in \eqref{defi.gama.pca}  
and we define the  corresponding sparsity parameters 
$$s_j := \|\gamma_j^0\|_0 , \;\;\;\text{ for }j=1,\dots,p.$$
These sparsity parameters as well correspond to the sparsity in the columns of $\Theta_0.$
To keep the presentation simpler, in the results that follow, we assume that the maximum eigenvalue of $\Sigma_0$ is bounded ($\phi_{\max} =\mathcal O(1)$) and that there exists a constant $c>0$ such that $\rho-3\eta \geq c>0$.
A more refined result might allow the quantities $\phi_{\max}$ and $1/(\rho-3\eta)$ to grow, although their growth cannot be faster than
 (a certain power of) ${\sqrt{n}}/({\max(s,\max_j s_j)\log p})$.

\begin{lemma}
\label{nodewise.asymp.spca}
Assume Condition \ref{design.spca} with a universal parameter $\sigma>0$, suppose that 
 $\rho-3\eta\geq c>0,$  $\phi_{\max} \leq C_{\max}$, 
for some universal constants $c,C_{\max}$  and $\max_{j=1,\dots,p}\spag=o(\sqrt{n/\log p})$.
Let $\hat\Theta$ be defined by the nodewise Lasso in Algorithm \ref{alg:nodewisea} with input matrix $\ddot R_n(\hat\beta)$, where
$\hat\beta$ is defined in \eqref{stationary} with suitable tuning parameters 
$$\lambda\asymp \sqrt{\log (2p)/n}, \;\;\;\text{ and }\;\;\;\|\beta_0\|_1\leq \tune\leq C_T \sqrt{n/\log(2p)},$$ and for $j=1,\dots,p$,
$$\lambda_j \asymp\sqrt{{\log (2p)}/{n}}\;\;\;\text{ and }
\;\;\; \|\gamma_j^0\|_1 \leq \tunej \leq \bar C_T \sqrt{{n}/{\log (2p)}},\;\;\;\;
$$
where $C_{\tune},\bar C_{T}$ are  suitable universal constants.
Then it holds
$$\max_{j=1,\dots,p}\|\hat\Theta_j-\Theta_j^0\|_1=\mathcal O_P(\max_{j=1,\dots,p}\spag \lambda_j),
$$
$$\max_{j=1,\dots,p}\|\hat\Theta_j-\Theta_j^0\|_2=\mathcal O_P(\max_{j=1,\dots,p}\sqrt{\spag} \lambda_j). $$
\end{lemma}


Our main results derive the asymptotic distribution of the entries $\hat b_j$ of  $\hat b$ and the asymptotic distribution of $\hat\Lambda$.

\begin{theorem}
\label{ci}
Assume Condition \ref{design.spca} holds with a universal parameter $\sigma$. Suppose that $\phi_{\max} \leq C_{\max}$
and $\rho-3\eta\geq c$ for some universal constants $C_{\max},c>0$.
Consider the estimator 
$$\hat b:= \hat\beta - \hat\Theta^T (\|\hat\beta\|_2^2 \hat\beta-\hat\Sigma \hat \beta  ),$$
 with $\hat\beta$ and $\hat\Theta$ as in Lemma \ref{nodewise.asymp.spca} and with the same tuning parameters as in Lemma \ref{nodewise.asymp.spca}.
Then, under the sparsity conditions 
$$ s=o(\sqrt{n}/\log p) \quad\text{ and }\quad\max_{j=1,\dots,p}s_j = o(\sqrt{n}/\log p) ,
$$
the de-sparsified estimator satisfies 
$$\hat b_{} -\beta_{0} = -\Theta_0 \dot R_n(\beta_0 ) + \emph{rem},$$
where 
 $$\|\emph{rem}\|_\infty =
\mathcal O_P\left(\max_{j=1,\dots,p} \max(s,s_j)\max\left(\lambda^2, \lambda_j^2,\frac{\log(2p)}{n}\right)\right)
=o_P\left(\frac{1}{\sqrt{n}}\right).
$$ 
\\
Moreover, for $j=1,\dots,p$, if $1/\sigma_j^2 =\mathcal O(1)$, it follows that
$$\sqrt{n}(\hat b_{j} -\beta^0_{j})/\sigma_j \rightsquigarrow \mathcal N(0,1),$$
where
$$\sigma_j^2:=
n\emph{var}((\Theta^0_j)^T\hat\Sigma\beta_0).$$
\end{theorem}
\vskip 0.25cm
We require sparsity  of small order  $\sqrt{n}/\log p$ in both  $\beta_0$ and in the columns of $\Theta_0$. We remark that for estimation of a single entry $\beta_j^0$, it is enough to assume sparsity in $\beta_0$ and in the corresponding column $\Theta_j^0.$
The sparsity requirement on $\beta_0$  is in line with literature on asymptotically normal estimation in sparse high-dimensional settings.
In particular, for linear regression, the same sparsity condition on the high-dimensional vector of regression coefficients is required. For linear regression, this condition was shown to be necessary for construction of confidence intervals (see \cite{cai.guo}).
A sparsity condition on the columns of the inverse Hessian of the population risk (here $\Theta_0$) also arises as a requirement for asymptotically normal estimation, see e.g. \cite{zhang}, \cite{vdgeer13}, \cite{stanford1}, \cite{vch1}.
Sparsity in the columns of $\Theta_0$ is for instance satisfied in the popular ``spiked covariance model'' (see e.g. \cite{jl}, \cite{montanari.pca}) as discussed in the example below.
\begin{example}
\label{ex}
In the spiked covariance model, the covariance matrix has the special form 
$$\Sigma_0 = I+\sum_{i=1}^r \omega_i u_i u_i^T,$$
 for $u_i,i=1,\dots,r$ being orthonormal vectors and $\omega_i$ positive numbers. Then one can easily deduce that the vectors $u_i$ are the first $r $  eigenvectors of $\Sigma_0$ with corresponding  eigenvalues $\Lambda_i = 1+\omega_i,i=1,\dots,r$.  
We denote the remaining $p-r$ eigenvectors by $u_i,i=r+1,\dots,p,$ and their eigenvalues are $\Lambda_i = 1$ for $i=r+1,\dots,p.$
Assuming $\omega_1 > \omega_2$, we have $\beta_0 = \sqrt{1+\omega_1}u_1$ and one can also deduce that the eigendecomposition of $\Theta_0$ is given by $U^T DU,$ where $U$ has rows $u_i,i=1,\dots,p$ and 
$D:=\text{diag}\left(
{2(1+\omega_1)} ,(\omega_1 -\omega_2), 
\dots, (\omega_1 -\omega_r), \omega_{1},\dots,
\omega_1\right)^{-1}.$
Then one can show  that 
\begin{eqnarray*}
\Theta_0 
&=&
  \sum_{i\leq r} D_{ii} u_iu_i^T + \frac{1}{\omega_1}(I-\sum_{i\leq r}u_iu_i^T)\\
	&=& 
	  \sum_{i\leq r} (D_{ii}-1/\omega_1) u_iu_i^T + \frac{1}{\omega_1}I. 
\end{eqnarray*}
If we assume that each of the first $r$ eigenvectors, $u_i,i=1,\dots,r$, has sparsity at most $s$, 
then each row of $\Theta_0$ has sparsity at most $rs+1.$
\end{example}

\noindent
For asymptotically normal estimation of the maximum eigenvalue, which is a quadratic function of $\beta_0$, we need to assume a somewhat stronger sparsity condition.
\begin{theorem}
\label{eigenvalue}
Assume the conditions of Theorem \ref{ci} and, in addition, assume that 
$$s^{3/2}=o(\sqrt{n}/\log p) \quad \text{ and } \max_{j=1,\dots,p}s_j^{3/2}=o(\sqrt{n}/\log p).$$
Recalling that 
$$\hat\Lambda = \|\hat \beta\|_2^2 - 2\hat\beta^T \hat\Theta^T \dot R_n(\hat\beta),$$
the following asymptotic expansion holds
\vskip 0.2cm
$$
\hat\Lambda - \Lambda_{\max}
 = -2\beta_0^T\Theta_0 \dot R_n(\beta_0) +\emph{rem},$$
\vskip 0.2cm
\noindent
where 
\begin{eqnarray*}
\|\emph{rem}\|_\infty &=&
\mathcal O_P\left(\max_{j=1,\dots,p}\max(s,s_j)^{3/2}\max\left(\lambda^2, \lambda_j^2,\frac{\log(2p)}{n}\right)\right)\\
&=&
o_P\left(\frac{1}{\sqrt{n}}\right).
\end{eqnarray*}
Denoting the variance of the pivot by
$$\sigma_{\Lambda}^2 := 4n\emph{var}(\beta_0^T\Theta_0 \hat\Sigma\beta_0),
$$
it follows
$$\sqrt{n}\left(\hat\Lambda - \Lambda_{\max} \right)/\sigma_{\Lambda} \rightsquigarrow \mathcal N(0,1).$$
\end{theorem}

The asymptotic variances of the  estimators in Theorems \ref{ci} and  \ref{eigenvalue} correspond to the asymptotic variance
of the loadings vector based on the sample covariance matrix from fixed-$p$-regime (see e.g. \cite{var.PC}).
For instance, if the observations are Gaussian $\mathcal N(0,\Sigma_0)$ and there are no eigenvalue multiplicities, then
$$\sigma_{\Lambda}^2=2\Lambda_{\max}^2,$$
\begin{eqnarray}
\label{asympvar}
\sigma_j^2 &= & (\Theta^0_{j})^T \Sigma_0 \Theta^0_j \|\beta_0\|_2^4 + [\|\beta_0\|_2^2(\Theta^0_j)^T \beta_0]^2,
\\\nonumber
  &=&\frac{\beta^0_j}{2}  +\|\beta_0\|_2^4 \sum_{i=1,i\not = j}^p u_{i,j}^2 \frac{\Lambda_{i}}{(\Lambda_{\max}-\Lambda_i )^2} ,
\end{eqnarray}
where $u_{i,j}$ is the $j$-th entry of the $i$-th eigenvector of $\Sigma_0$
 (see Lemma \ref{variance} in Section \ref{sec:proof.normal} for the derivation).
The asymptotic variance $\sigma_{\Lambda}^2$ can be easily estimated by $\hat\sigma_{\Lambda}^2:= 2\|\hat\beta\|_2^2.$
The asymptotic variances $\sigma_j^2$ however depend on all the eigenvectors $u_i,i\not =j$. Simultaneous estimation of all the eigenvectors would in theory require $p\ll n$ and is moreover impractical if we are only interested in inference about the first few loadings vectors.
Therefore, we suggest to use an alternative procedure, which  computes the natural estimator
\begin{equation}
\label{naturale}
\hat\sigma_j^2 := \frac{1}{n}\sum_{i=1}^n (\hat\Theta_j^T \Xobsi(\Xobsi)^T\hat\beta)^2 - (\hat\Theta_j^T \hat\Sigma\hat\beta)^2.
\end{equation}
This does not assume the knowledge of the distribution of $\Xobsi$ and is only based on the estimators $\hat\Sigma, \hat\Theta_j$ and $\hat\beta.$ Analogously one can estimate the variance $\sigma_{\Lambda}^2$ in the non-Gaussian case.
We omit the theoretical guarantees for these estimators, but point the reader to results of a similar flavour which are proved for estimation of asymptotic variance 
in \cite{jvdgeer14}
 under sub-Gaussianity conditions on the design. 
\par
The result of Theorem \ref{ci} can be applied for support recovery of the entries of $\beta_0$ by thresholding the 
de-sparsified estimator at the level $C\sqrt{\log p/n}$ for a suitable (possibly data-driven) $C>0.$ 
Define the thresholded estimator 
$$\hat b_{\text{thresh},i}:= \hat b_i 1_{\hat b_i>C\sqrt{\frac{\log p}{n}}},$$
where $\hat b_i$ is the $i$-th entry of $\hat b.$ 
Then $\hat b_{\text{thresh}}:=(\hat b_{\text{thresh},1,},\dots,\hat b_{\text{thresh},p})$ recovers no false positives, i.e. the support $\hat S$ of $\hat b_{\text{thresh}}$ satisfies asymptotically, with probability tending to one,
$$\hat S \subseteq S_0,$$
where $S_0$ is the support of $\beta_0.$
Moreover, if in addition the beta-min condition holds, i.e.
 $$\min_{j\in S_0}\beta^{0}_j \geq 2C\sqrt{\frac{\log p}{n}},$$ 
then we obtain exact support recovery, i.e. $\hat S =S_0$ (asymptotically, with probability tending to one).
The problem of support recovery of the first eigenvector was studied in a number of papers, see  e.g. \cite{jl}, \cite{amini.wainwright}, \cite{montanari.pca}, under irrepresentability conditions or under the spiked covariance model. Our results do not need to assume the irrepresentability condition, but 
require a sparsity condition on $\Theta_0$, which may be viewed as a less stringent condition.









\section{Empirical results}
\label{sec:sim}

\vskip 0.3cm
\par
\subsection{Setup}
In this section, we demonstrate the performance of the de-biased sparse PCA in several models and different dimensionality regimes. We provide a comparison to the classical PCA.
  
We consider the spiked covariance model with a single spike,
$$\Sigma_0= I+\omega vv^T,$$
where $$v=(1,1,1,0,1,0,\dots,0) \in \mathbb R^p$$
for two different spike sizes $\omega$:

\vskip 0.2cm
\begin{itemize}
\item
\textbf{Model 1 (Small spike): }$\omega=1/5,$
\item
\textbf{Model 2 (Large spike): }$\omega=1.$
\end{itemize}

\vskip 0.2cm
\noindent
The observations $\Xobso,\dots,\Xobsn$ are independent and $\mathcal N(0,\Sigma_0)$-distributed.

\vskip 0.2cm
\noindent
The more challenging model is arguably 
Model 1, where the eigenvalue gap is smaller. Indeed, one can easily check that 
for Model 1,  $\Lambda_{\max} = 1.8$, while for Model 2, $\Lambda_{\max} = 5$. 
 As we will see in the simulation study, classical PCA does not perform well in Model 1, while it does perform well in our Model 2.
In terms of theoretical conditions such as sparsity, one can check that the vector $\beta_0 = \sqrt{1+\|v\|_2^2}v$ is the first loadings vector with sparsity $s=4$  and 
the inverse Fisher information $\ddot R(\beta_0)^{-1}$ has sparsity $4.$

\vskip 0.1cm \noindent
We demonstrate the performance of the de-biased sparse PCA  (and classical PCA) for construction of confidence intervals for individual entries of $\beta_0$.
We first calculate the confidence intervals assuming the asymptotic variance from \eqref{asympvar} is known. This gives a fairer comparison, otherwise for the classical PCA, we would observe that too large estimates of asymptotic variance lead to large confidence intervals and perfect coverage. We look at estimating the asymptotic variance separately. 

The sparse PCA estimator \eqref{est.fast} is calculated using gradient descent, with a tuning parameter $\lambda= \sqrt{\log p/n}$ and the starting point of the algorithm is the initial estimator $\hat\beta_{\text{init}}$.
The constraint on the $\ell_1$-norm turns out to be unnecessary in our simulations.
We compute the non-convex nodewise Lasso estimator with tuning parameters $\lambda_j = \sqrt{\log p/n},j=1,\dots,p$.

For Model 1, we investigate the scenarios:
 $(p=200,n=200)$ (Figure \ref{fig:m122}), $(p=200,n=400)$ (Figure \ref{fig:m124}) and $(p=500,n=800)$
 (Figure \ref{fig:m155}). 
For Model 2, we consider the scenario $(p=200,n=200)$ (Figure \ref{fig:m222}).
The target coverage is $95\%$ in all simulations.
The average coverage is reported over the non-zero set 
$S_0:=\{i:\beta_i^0 \not = 0\}$ and $S_0^c$.
The number of generated random samples is always $N=200$.

We can observe that the classical PCA does not perform well in estimation of the non-zero entries of $\beta_0$ in Model 1, while the de-biased estimator performs reasonably well.
We also find that our  theoretical condition requiring $s=o(\sqrt{n}/\log p)$ seems to be \emph{needed} for our method to perform well in simulations. Namely, comparing Figures \ref{fig:m122} and \ref{fig:m124}, we see that the performance of our estimator was substantially improved with the increased sample size. Note that in the setting in Figure \ref{fig:m122} $(p=200,n=200)$,  we have sparsity $s=4$ and $\sqrt{n}/\log p  \approx 2.67$, while in Figure \ref{fig:m124}  $(p=200,n=400)$, we still have sparsity $s=4$ but due to a bigger sample size, we have $\sqrt{n}/\log p\approx 3.77.$ This confirms our theoretical findings and we note that a similar phenomenon has also been observed in other settings: the generalized linear models in \cite{anor} and Gaussian graphical models (\cite{jvdgeer14} and \cite{jvdgeer15}). 

\par
Finally, we look at estimating the asymptotic variance, measured by the  length of confidence intervals given by
$\Phi^{-1}(0.95)\frac{\hat\sigma_j }{\sqrt{n}},$
where $\hat\sigma_j^2$ is the estimator of asymptotic variance estimator as proposed in \eqref{naturale}. 
For the de-biased sparse PCA, we use $\hat\beta$ and $\hat\Theta$ as defined in Section \ref{sec:spca.main} to calculate the estimate of the asymptotic variance \eqref{naturale}. For the classical PCA, we use $\hat\beta:=\hat\beta_{PCA}$ and $\hat\Theta:=\hat\Sigma^{-1}$ to calculate  \eqref{naturale}. 
The results are reported in Table \ref{tab:av}. The average length of a confidence interval is calculated over $N=100$ randomly generated samples. We also report the ``Asymptotically efficient length'', which is the asymptotically optimal length of a confidence interval corresponding to the fixed-$p$ setting.

\section{Discussion}
\label{sec:disc}
We have proposed a computationally feasible methodology with theoretical guarantees for constructing confidence intervals for loadings and the maximum eigenvalue of the covariance matrix in a sparse high-dimensional regime. 
The results may also be applied for support recovery without requiring irrepresentability conditions, 
although we do require the (arguably weaker) sparsity condition on the columns of the inverse population Hessian matrix.
We have shown that the de-biasing methodology which was studied in a line of papers (\cite{zhang,vdgeer13,jvdgeer14,jvdgeer15}) may be used even in a non-convex setting.
The challenge here lied especially in estimating the inverse Fisher information, which is not guaranteed to be positive definite under non-convexity of the loss function.
\par
To position our research relative to the existing literature on asymptotic normality for principal component analysis in high dimensions, 
it is worth to point out that contrary to the papers \cite{koltchinskii2016asymptotics} and \cite{fan.spca}, our results do not study the special setting where the maximum eigenvalue diverges, or where the eigenvalue gap diverges. 
We allow the eigenvalue gap to be very small,
what arguably presents a more challenging setting, requiring us to rely on sparsity conditions.

\begin{figure}[h!]
\centering
\begin{center}
\begin{tabular}{cc}
\multicolumn{2}{c}{\bf Model 1: p = 200, n = 200}\\[0.2cm]
 De-biased sparse PCA & Classical PCA\\
\centering
\includegraphics[width=0.47\textwidth]{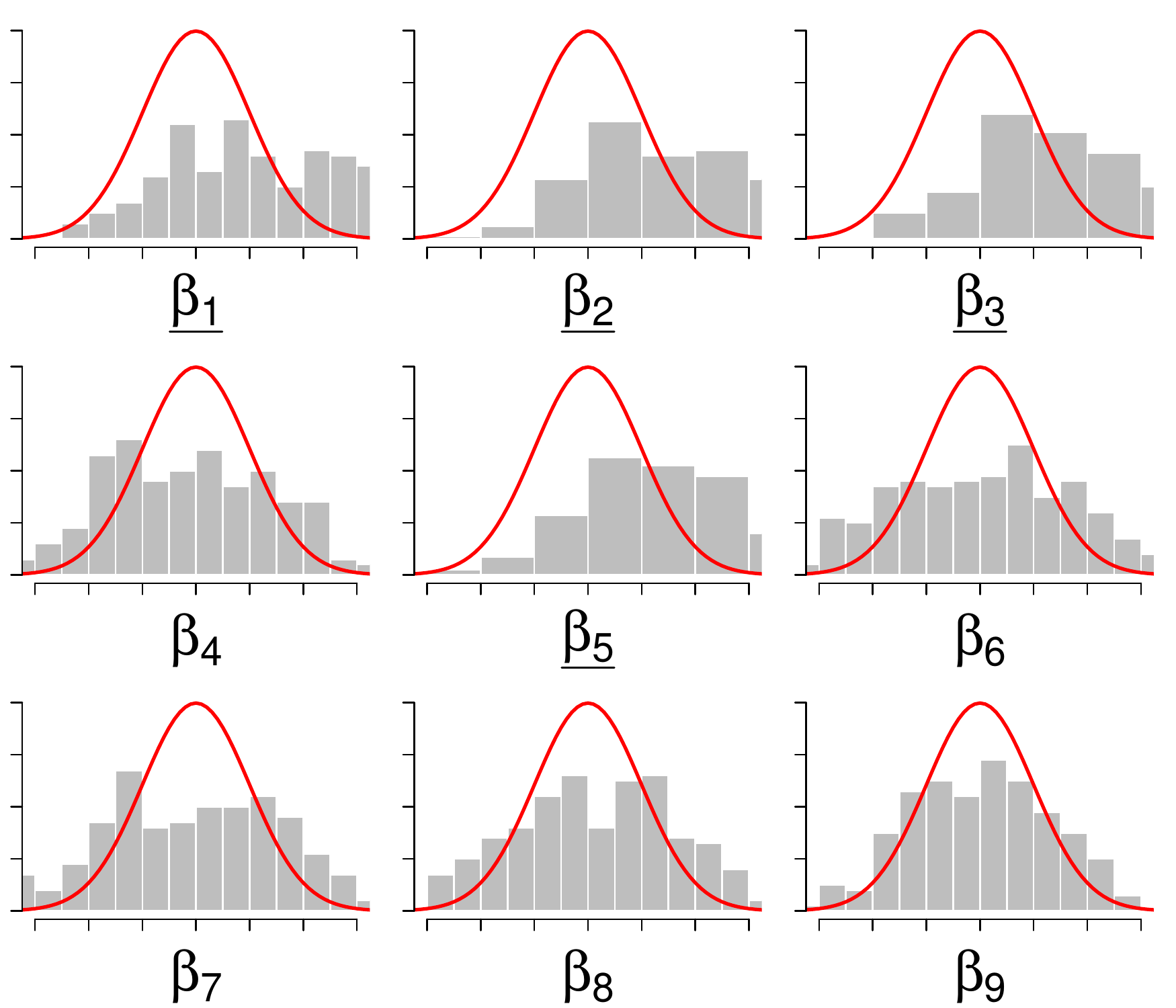}
& \includegraphics[width=0.47\textwidth]{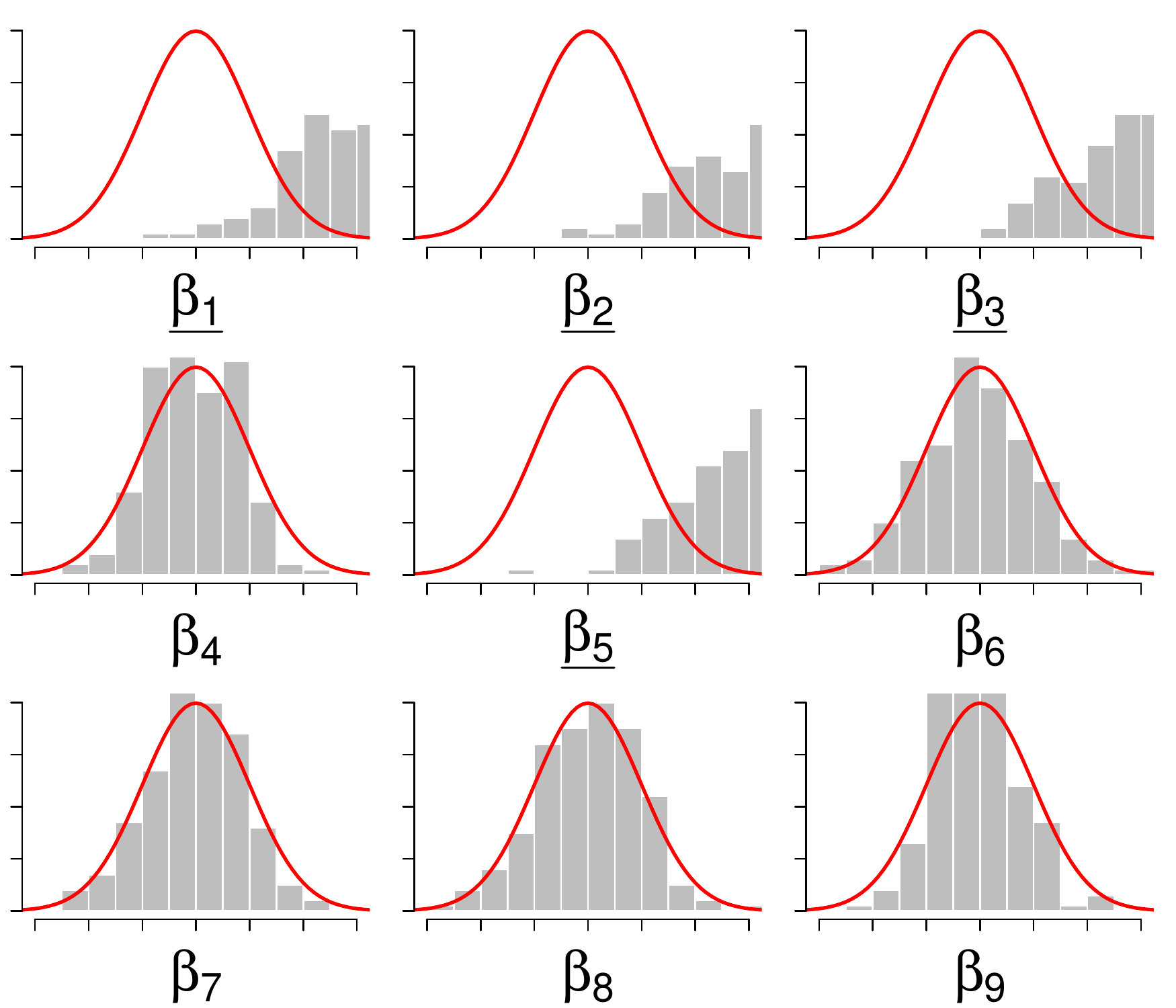} 
\end{tabular}
\end{center}
\small
\begin{tabular}{ccc}
  \hline
	&\multicolumn{2}{c}{\bf Average coverage} 
	\\
 {\bf Method } &  $S_0$ & $S_0^c$ 
	\\ 
  \hline
De-biased sparse PCA & 0.78 & 0.84 
\\ 
  Classical PCA & 0.16 & 0.98 
	\\ 
   \hline
\end{tabular}
\caption[Histograms for the loadings.]{
Histograms corresponding to (normalized) estimators of the first 9 entries of the loadings vector $\beta_0$.
Left:  de-biased sparse PCA estimator, right:  classical PCA. The non-zero entries of $\beta_0$ are \underline{underlined}. 
}
\label{fig:m122}

\end{figure}

\begin{figure}[h!]
\centering
\begin{tabular}{cc}
\multicolumn{2}{c}{\bf Model 1: p = 200, n = 400}\\[0.2cm]
 De-biased sparse PCA & Classical PCA\\
\centering
\includegraphics[width=0.47\textwidth]{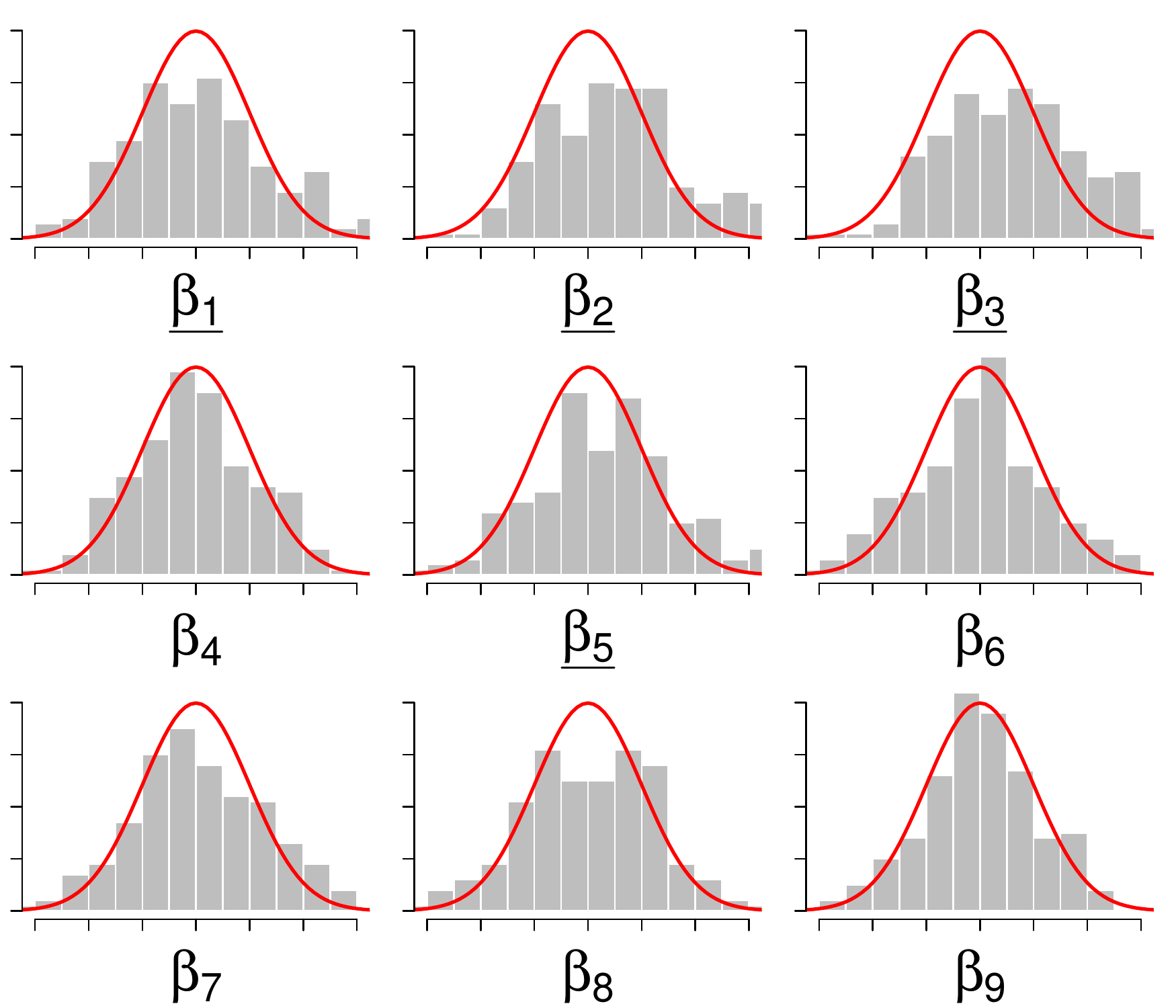}
& \includegraphics[width=0.47\textwidth]{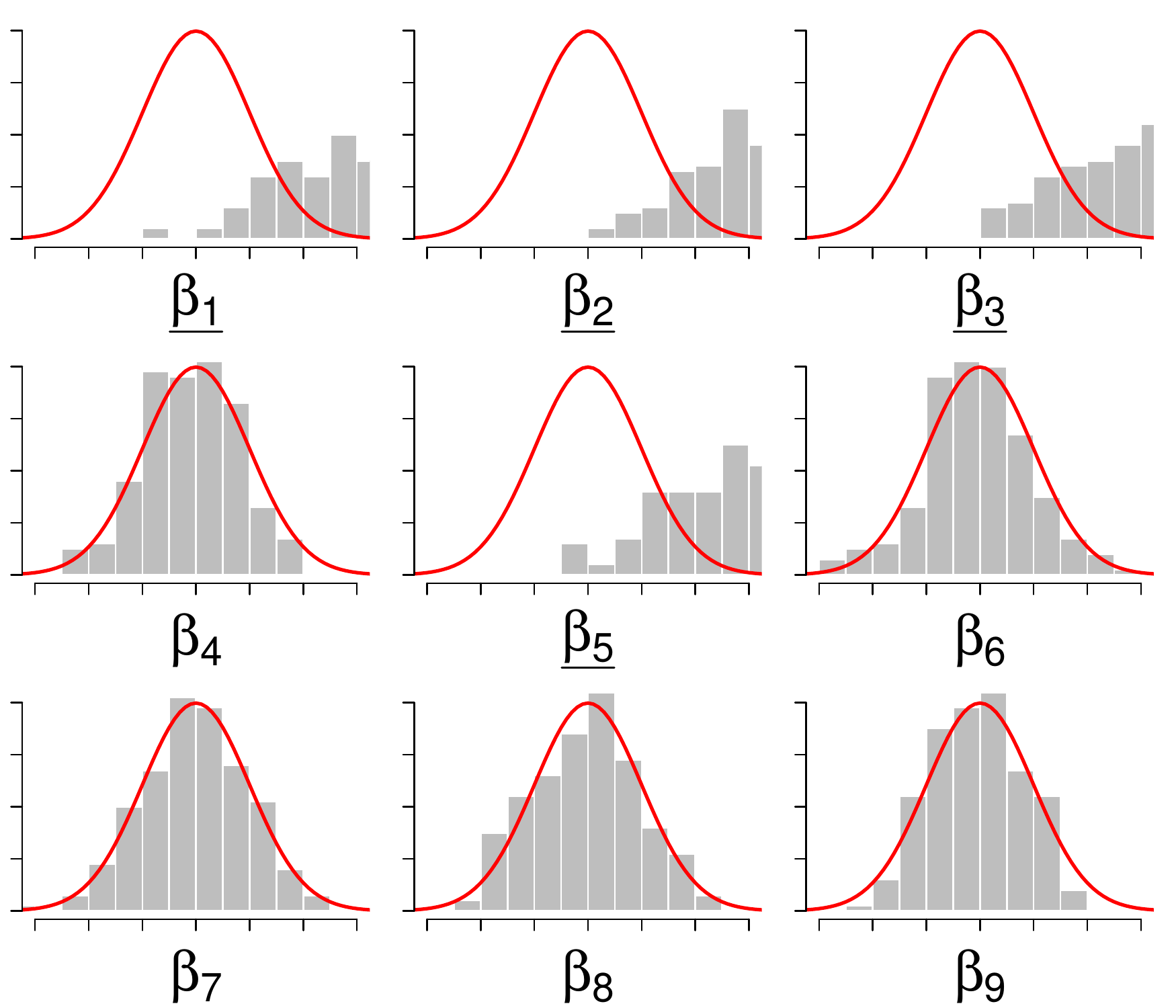} 
\end{tabular}
\small
\begin{tabular}{ccc}
  \hline
 & \multicolumn{2}{c}{\bf Average coverage} \\
{\bf Method }   &  $S_0$ & $S_0^c$ 	\\
  \hline
De-biased sparse PCA& 0.95 & 0.97  \\ 
  Classical PCA  & 0.24 & 0.96  \\ 
   \hline
\end{tabular}
\caption{
Histograms corresponding to (normalized) estimators of  $\beta_0$.
}
\label{fig:m124}
\end{figure}

%


\begin{figure}[h!]
\centering
\begin{tabular}{cc}
\multicolumn{2}{c}{\bf Model 1: p = 500, n = 800}\\[0.2cm]
 De-biased sparse PCA & Classical PCA\\
\centering
\includegraphics[width=0.47\textwidth]{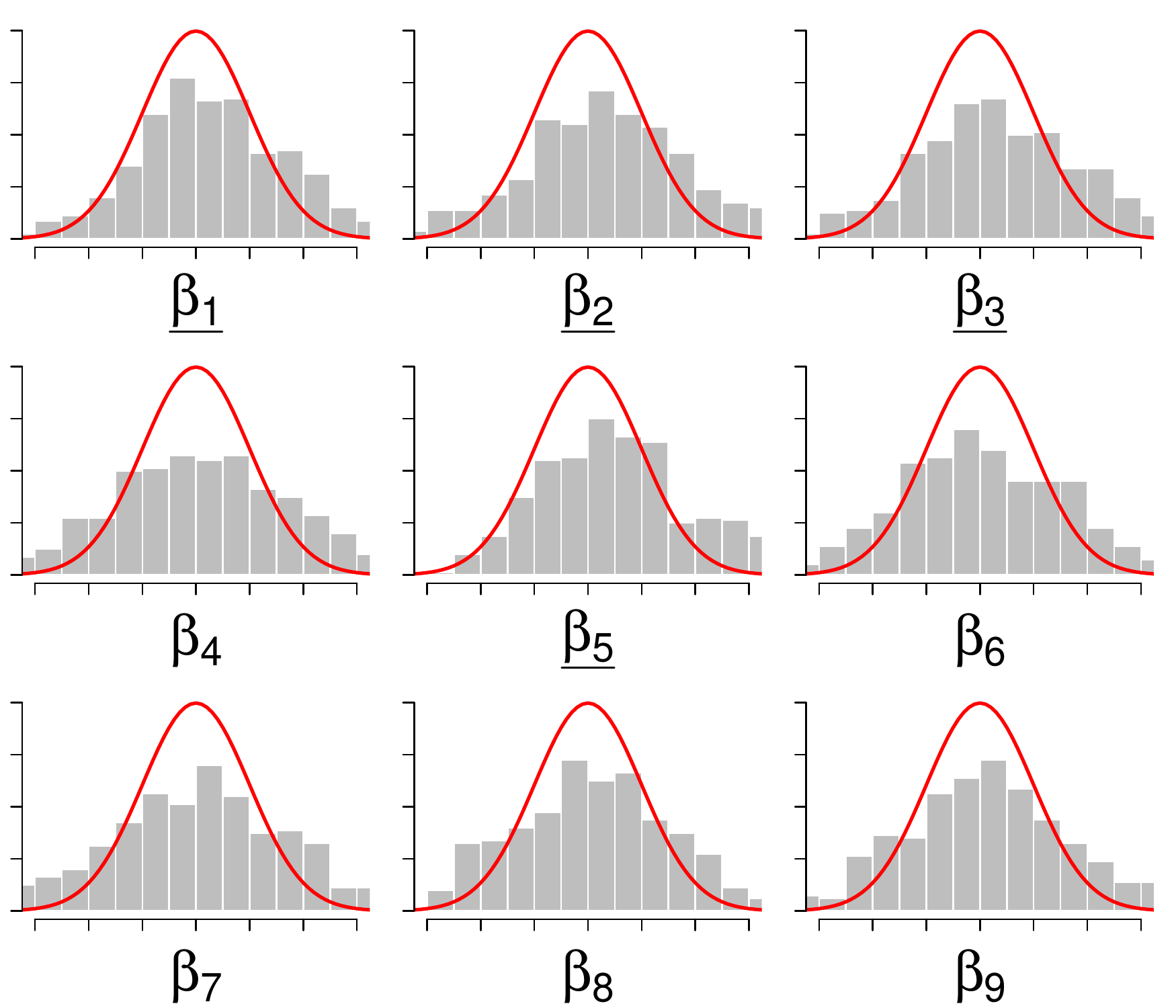}
& \includegraphics[width=0.47\textwidth]{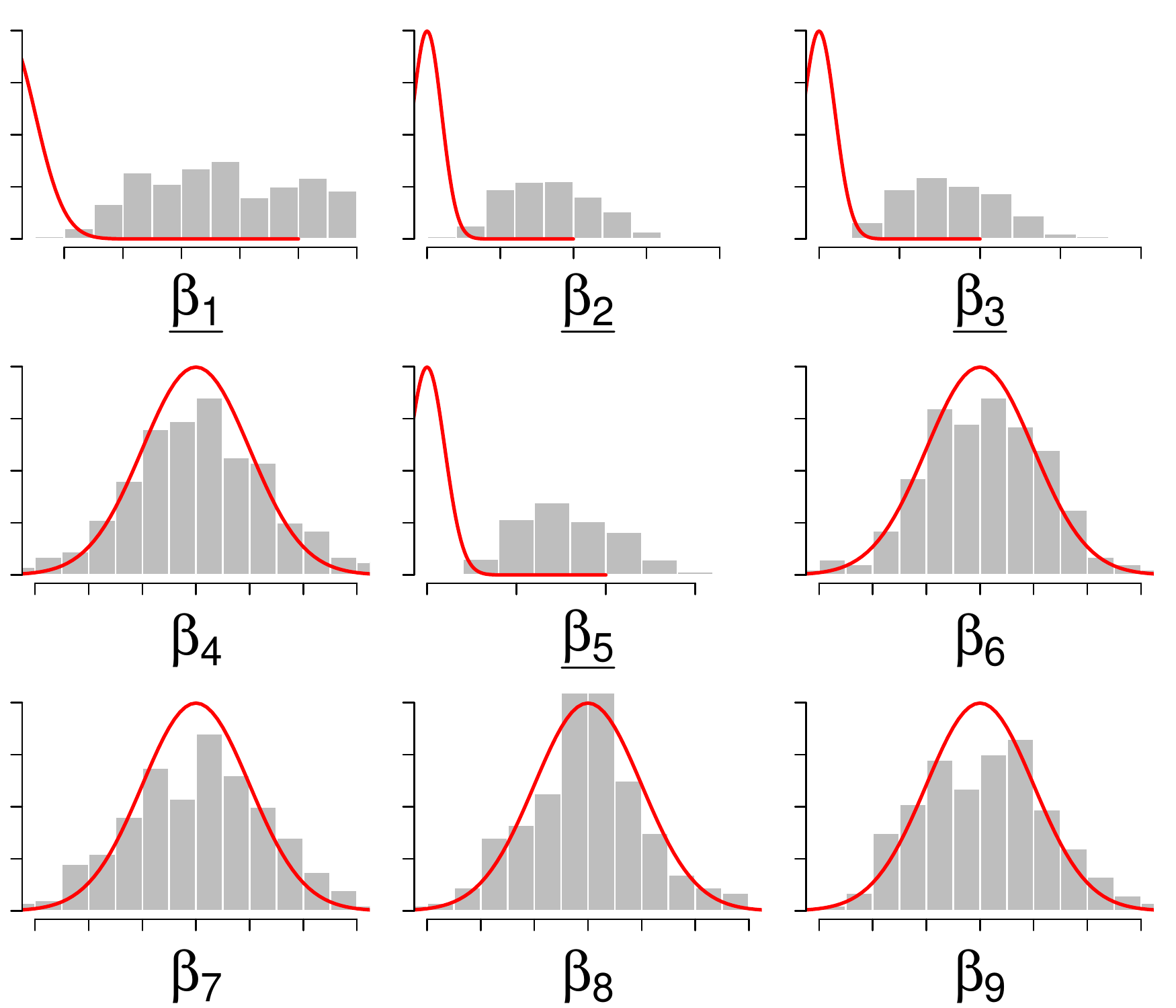} 
\end{tabular}
\small
\begin{tabular}{rrrrr}
  \hline
		&\multicolumn{2}{c}{\bf Average coverage} 
	\\
 {\bf Method }  &  $S_0$ & $S_0^c$ 
	\\ 
  \hline
De-biased sparse PCA & 0.78 & 0.77\\ 
Classical PCA & 0.00 & 0.89  \\ 
   \hline
\end{tabular}
\caption{Histograms corresponding to (normalized) estimators of the first 9 entries of  $\beta_0$.
}
\label{fig:m155}
\end{figure}


\begin{figure}[h!]
\centering
\begin{tabular}{cc}
\multicolumn{2}{c}{\bf Model 2: p = 200, n = 200}\\[0.2cm]
 De-biased sparse PCA & Classical PCA\\
\centering
\includegraphics[width=0.47\textwidth]{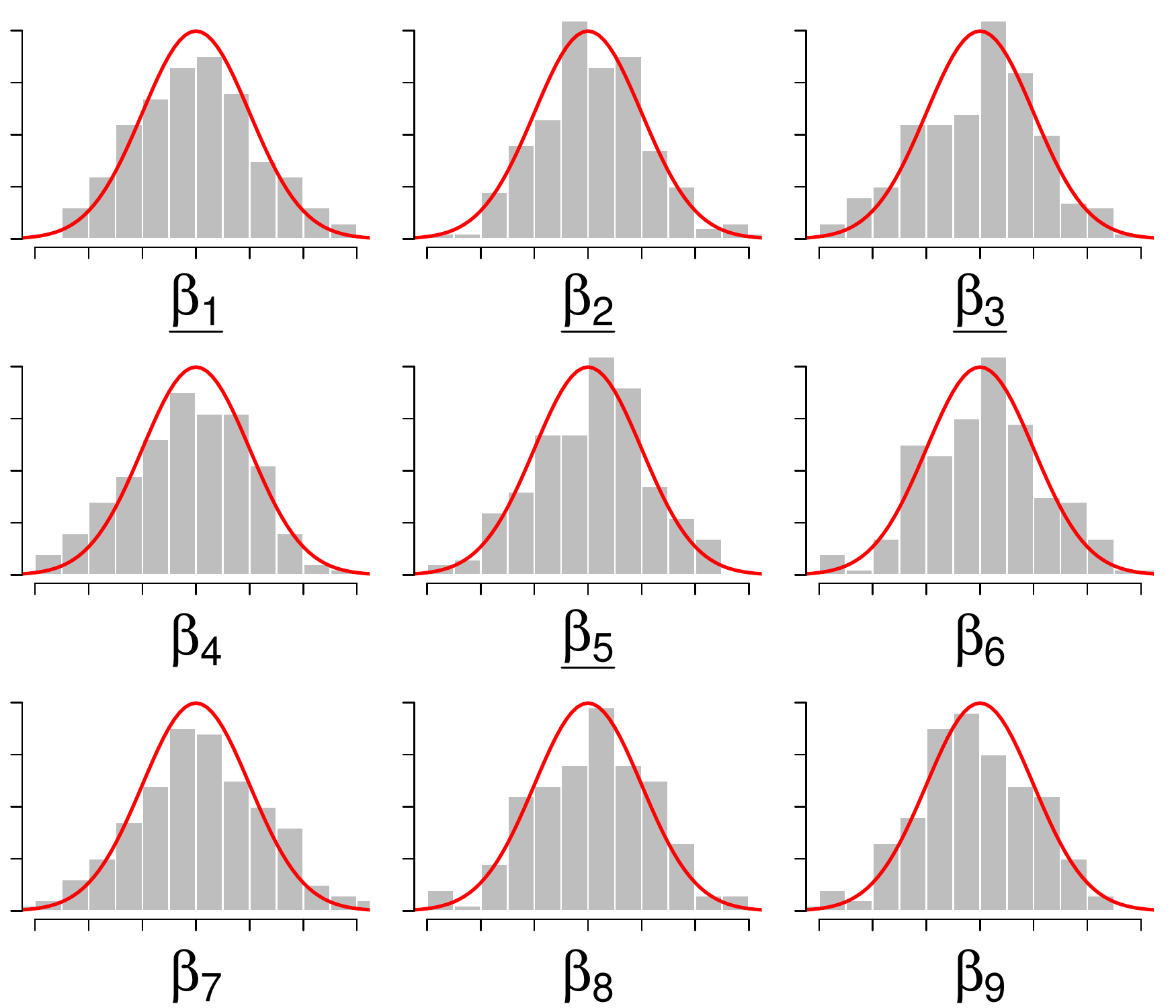}
& \includegraphics[width=0.47\textwidth]{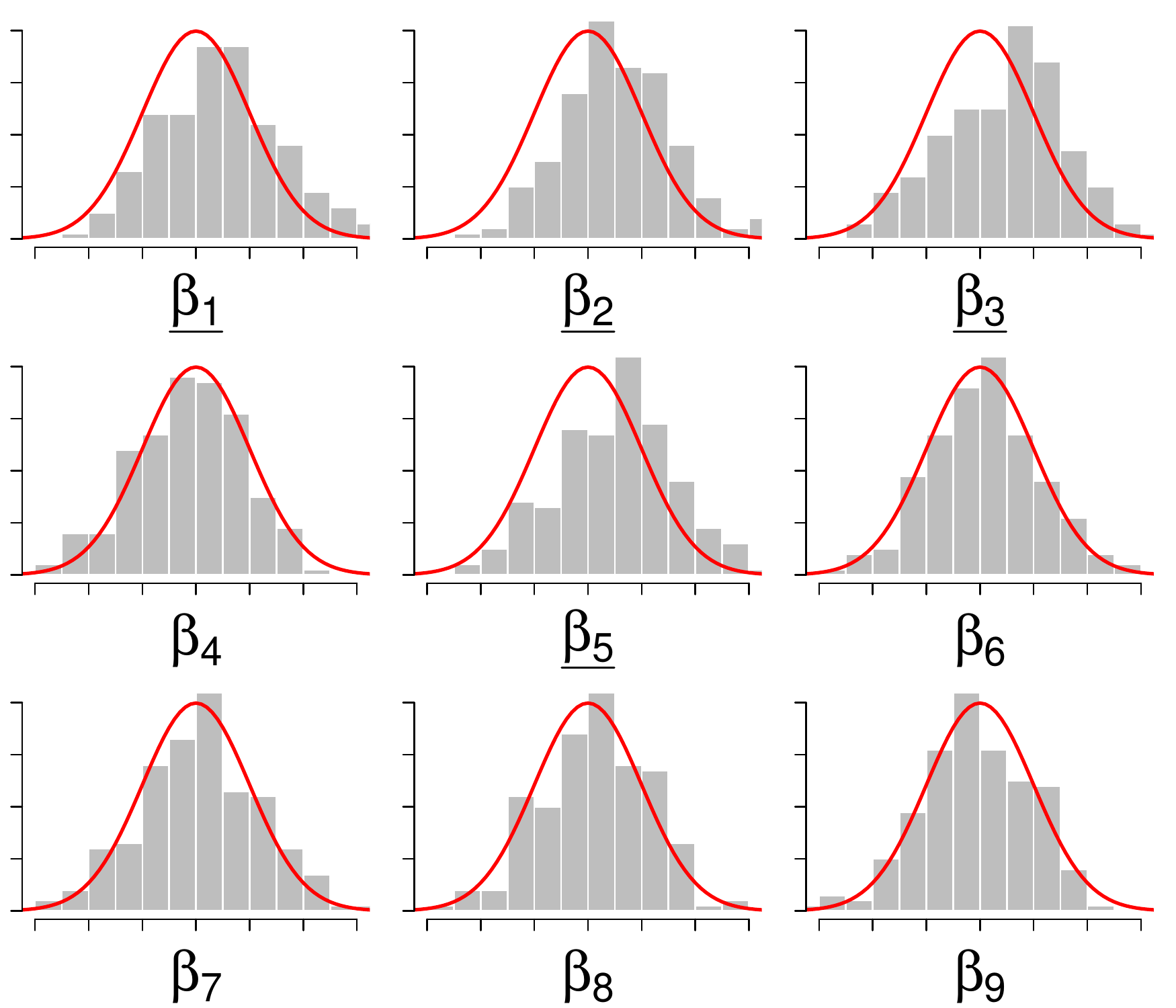} 
\end{tabular}
\small
\begin{tabular}{rrrrr}
  \hline
		&\multicolumn{2}{c}{\bf Average coverage} 
	\\
 {\bf Method }  &  $S_0$ & $S_0^c$ 
	\\ 
  \hline
De-biased sparse PCA & 0.96 & 0.93  \\ 
Classical PCA & 0.94 & 0.95  \\ 
   \hline
\end{tabular}
\caption{Histograms corresponding to (normalized) estimators of the first 9 entries of $\beta_0$.
}
\label{fig:m222}
\end{figure}

\begin{table}[h!]
\centering
\textbf{Estimating the asymptotic variance}
\footnotesize
\vskip 0.4cm
\begin{tabular}{lrrrr}
  \hline
 \multirow{2}{*}{\bf Model 1: p = 200, n = 200} & 
 \multicolumn{2}{c}{Average length}  \\ 
	& $S_0$ & $S_0^c$ \\
  \hline\\[-5pt]
De-biased sparse PCA  
& 0.406 & 0.319 \\ 
Classical PCA  
& 3.327 & 3.496 \\ 
 Asymptotically efficient length$^*$ 
& 0.278 & 0.312 \\ 
   \hline\\[3pt] 
\end{tabular}
\vskip 0.3cm
%
\centering
\footnotesize
\begin{tabular}{lrrrr}
  \hline
 \multirow{2}{*}{\bf Model 2: p = 200, n = 200 }& 
 \multicolumn{2}{c}{Average length}  \\ 
	& $S_0$ & $S_0^c$ \\
  \hline \\[-5pt]
De-biased sparse PCA 
& 0.178 & 0.181 \\ 
  Classical PCA 
	& 0.232 & 0.268 \\ 
  Asymptotically efficient length$^*$
	& 0.186 & 0.173 \\[3pt] 
   \hline
\end{tabular}
\caption[Estimating asymptotic variance.]{Average length of confidence intervals. \\
$^*$ corresponding to the fixed-p regime (see \cite{var.PC}).}
\label{tab:av}
\end{table}

\section{Proofs}
\label{sec:proofs}

\ifsubsection
\section{Proofs for Section \ref{sec:spca.prelim}: First step estimator }
\else
\subsection{Proofs for Section \ref{sec:spca.prelim}: First step estimator }
\fi

\label{subsec:oracle}

\begin{proof}[Proof of Lemma \ref{spca.init}]
Using the arguments of Theorem 3.3 in \cite{fantope}, one can easily show (with the techniques used to prove oracle inequalities for $\ell_1$-regularized estimators - see e.g. \cite{hds}) that for $\lambda\geq 2\|\hat\Sigma-\Sigma_0\|_\infty$, 
$$\|\hat Z - \eig \eig^T \|_F^2 + \lambda \|\hat Z - \eig\eig^T\|_1 \leq \frac{C}{(\Lambda_1-\Lambda_2)^2}s^2\lambda^2,$$
where $C$ is a universal constant. 
Let $\|A\|:=\sqrt{\Lambda_{\max}(AA^T)}$ denote the spectral norm.
Since for any square matrix $A$ it holds
$\|A\|\leq
\|A\|_F $,
then
$$
\|\hat Z - \eig \eig^T \| \leq \|\hat Z - \eig \eig^T\|_F .$$
Hence 
$$
\|\hat Z\| \geq \|\eig\eig^T\| - \epsilon =1  - \epsilon.$$
Write the eigendecomposition of $\hat Z$ as
$$\hat Z = \sum_{j=1}^p \hat\phi_j^2 \hat u_j\hat u_j^T,$$
where $\hat\phi_1 \geq \dots \geq \hat\phi_p$ and $\hat u_j^T \hat u_j = 1. $
Since  $\hat \phi_1^2 = \|\hat Z\|\geq 1-\epsilon$ and $\text{tr}(\hat Z) = \sum_{j=1}^p \hat\phi_j^2=1$,
then 
$$\text{tr}(\hat Z^T \hat Z)=\text{tr}(\hat Z^2) = \sum_{j=1}^p \hat\phi_j^4=\hat\phi_1^4 + \sum_{j=2}^p\hat \phi_j^4\geq (1-\epsilon)^2.$$
Hence it follows
\begin{eqnarray*}
\|\hat Z - \eig\eig^T\|_F^2 &=&
 2 \left(1 - \text{tr}(\hat Z \eig\eig^T)\right) + \text{tr}(\hat Z^T\hat Z)-1 \\
&\geq &
 2 \left(1 - \text{tr}(\hat Z \eig\eig^T)\right) - 2\epsilon + \epsilon^2 \\
&=&
2 \left(1 - \sum_{j=1}^p\hat \phi_j^2 (u_1^T \hat u_j)^2\right)- 2\epsilon + \epsilon^2.
\end{eqnarray*}
Thus
\begin{equation}
\label{pcae1}
2 \left(1 - \sum_{j=1}^p \hat\phi_j^2 (u_1^T \hat u_j)^2\right)\leq 2\epsilon.
\end{equation}
Moreover,
\begin{equation}
\label{pcae2}
\sum_{j=1}^p \phi_j^2 (u_1^T u_j)^2 \leq \phi_1^2 (\eig^T u_1)^2 + \sum_{j=2}^p \phi_j^2 \leq (\eig^T u_1)^2+\epsilon.
\end{equation}
Then combining \eqref{pcae1} and \eqref{pcae2} it follows
$$2 \left(1 - (\eig^T \hat u_1)^2\right)\leq 4\epsilon.$$
Since we assume without loss of generality that $\hat u_1^T \eig \geq 0,$
$$\|\eig - \hat u_1\|_2^2 = 2(1-\eig^T \hat u_1) = \frac{2(1-(\eig^T\hat u_1)^2)}{1+\eig^T\hat u_1} \leq 4\epsilon.$$
Now we proceed to show the bound for $\|\hat \beta_{\emph{init}}-\beta_0\|_2.$ Recall that $\hat \beta_{\emph{init}}= \text{tr}(\hat\Sigma\hat Z)^{1/2}\hat u_1.$ Using the eigendecomposition of $\hat Z$, we can write
\begin{eqnarray*}
|\text{tr}(\hat\Sigma \hat Z) - \text{tr}(\Sigma_0 \eig\eig^T) |
& \leq  &
\underbrace{|\text{tr}((\hat\Sigma-\Sigma_0)\hat Z)|}_{i_1}
 + \underbrace{|\text{tr}(\Sigma_0 (\sum_{j=2}^p \hat\phi_j^2 \hat u_j \hat u_j^T))|}_{i_2}\\
&&\;
+\underbrace{|\text{tr}(\Sigma_0 (\hat\phi_1^2\hat u_1 \hat u_1 - \eig\eig^T)|}_{i_3}
.
\end{eqnarray*}
Firstly,  since $\|u_1u_1^T\|_1\leq s\|u_1u_1^T\|_F\leq s$ and $\|\hat Z- u_1u_1^T\|_1\leq \epsilon^2/\lambda$, it follows
$$i_1 \leq\|\hat\Sigma-\Sigma_0\|_\infty \|\hat Z\|_1 \leq \lambda/2 (s + \epsilon^2/\lambda) .$$
Secondly,  
$$i_2 =\text{tr}(\Sigma_0 (\sum_{j=2}^p \hat\phi_j^2 \hat u_j\hat u_j^T )) =
\sum_{j=2}^p \hat\phi_j^2 \hat u_j^T\Sigma_0\hat u_j  \leq 
\Lambda_{\max} \sum_{j=2}^p \hat\phi_j^2  \leq  \Lambda_{\max} \epsilon.$$
Thirdly, 
\begin{eqnarray*}
i_3 =|\text{tr}(\Sigma_0 (\hat\phi_1^2\hat u_1 \hat u_1^T - \eig\eig^T)|
&\leq & |\text{tr}(\Sigma_0 (\hat\phi_1^2-1)\hat u_1 \hat u_1^T)|
\\
&&\;\;+\;\;
|\text{tr}(\Sigma_0 (\hat u_1 \hat u_1^T - \eig\eig^T)|
\\[5pt]
&\leq & \epsilon \|\hat u_1\|_2 \Lambda_{\max}  + 
|(\hat u_1 -u_1)^T\Sigma_0(\hat u_1 -\eig)| \\
&&\;\;+\;\; 2\;|u_1^T \Sigma_0 (\hat u_1 -\eig)| 
\\[5pt]
&\leq &
\epsilon \Lambda_{\max} +
\Lambda_{\max} \|\hat u_1 -\eig\|_2^2\\
&&\;\; +\; \;2 \Lambda_{\max}^{1/2}\|\hat u_1 -\eig\|_2
\\[5pt]
&\leq &
5 \epsilon \Lambda_{\max}   +  4\sqrt{\epsilon} \Lambda_{\max}^{1/2}
.
\end{eqnarray*}
Hence, collecting the bounds,
\begin{eqnarray*}
|\text{tr}(\hat\Sigma \hat Z) - \text{tr}(\Sigma_0 \eig\eig^T) |
& \leq  &
\lambda/2 (s + \epsilon^2 /\lambda) + 2\left(3\Lambda_{\max}  \epsilon  +  2\Lambda_{\max}^{1/2} \sqrt{\epsilon}\right) \\
&:=&\;\; \zeta
.
\end{eqnarray*}
But then, assuming $\Lambda_{\max}-\zeta>0$,
\begin{eqnarray*}
\|\hat\beta_{\text{init}}-\beta_0\|_2 &\leq &
 |\text{tr}(\hat\Sigma \hat Z)^{1/2}- \text{tr}(\Sigma_0 \eig\eig^T)^{1/2}|\|\hat u_1\|_2 \\
&&+ \;\text{tr}(\Sigma_0 \eig\eig^T)^{1/2}\| \hat u_1 - u_1\|_2
\\[5pt]
&=&
\frac{1}{2\sqrt{\Lambda_{\max}-\zeta}}\zeta + 2\sqrt{\epsilon}\Lambda_{\max}^{1/2}
.
\end{eqnarray*}

\end{proof}

\ifsubsection
\section{Proofs for Section \ref{subsec:spca.theory}}
\else
\subsection{Proofs for Section \ref{subsec:spca.theory} }
\fi

\subsubsection{Oracle inequalities for the second step estimator}

\begin{proof}[Proof of Theorem \ref{oracle}] 
The definition of a stationary point $\betatil$ in particular implies
\begin{equation}
\label{kkt1234}
(\dot R_n(\betatil) + \lambda \hat Z)^T (\beta_0-\betatil) \geq 0,
\end{equation}
where $\hat Z$ is the sub-differential of the $\ell_1$-norm of $\beta$ evaluated at $\betatil.$
By Taylor expansion of the population loss, we obtain
\begin{equation}
\label{taylor1234}
R(\beta_0) - R(\betatil) = \dot R(\betatil)^T(\beta_0-\betatil) + \frac{1}{2}(\beta_0-\betatil)^T\ddot R(\bar \beta) (\beta_0-\betatil),
\end{equation}
for an intermediate point $\bar \beta=\alpha \hat\beta+(1-\alpha)\beta_0$, for some $\alpha\in [0,1].$\\
Since $\|\bar\beta-\beta_0\|_2\leq \eta,$ Lemma \ref{eigen.spca} implies $\Lambda_{\min}(\ddot R(\bar \beta)) \geq 2(\rho-3\eta)>0.$
Thus, combining \eqref{kkt1234} and \eqref{taylor1234} and rearranging yields

\begin{equation*}
\label{four}
R(\betatil) -R(\beta_0) 
+ (\dot R_n(\betatil) - \dot R(\betatil))^T(\beta_0-\betatil) -
 \lambda \hat Z^T (\beta_0-\betatil) \leq 0.
\end{equation*}

\noindent
Using that $\hat\beta$ and $\hat Z$ satisfies
$$\hat Z^T \betatil = \|\betatil\|_1\;\;\text{ and }\;\;|\hat Z^T \beta_0|\leq \|\hat Z\|_\infty\|\beta_0\|_1\leq \|\beta_0\|_1,$$
it follows
\begin{equation*}
\label{four}
R(\betatil)-R(\beta_0)  
+ \lambda \|\betatil\|_1 \quad \leq \quad \lambda\|\beta_0\|_1+ \mathcal E(\hat\beta),
\end{equation*}
where we denoted the empirical process term by 
$$\mathcal E(\beta):=|(\dot R_n(\beta) - \dot R(\beta))^T(\beta_0-\beta)|.$$
\noindent
It remains to bound the random term 
$\mathcal E(\betatil)$. Note that 
$$\mathcal E(\betatil) = -(\hat\beta -\beta_0)^T W(\hat\beta-\beta_0) +
\beta_0^T W(\beta_0-\hat\beta),$$
where we denote $W:=\hat\Sigma-\Sigma_0.$
First note that for $\lambda_1 = 4\sigma^2\|\beta_0\|_2(\lambda_0 + \lambda_0^2),$ by Lemma \ref{sigma.con} it follows that
with probability at least $1-\alpha$, where $\alpha:=2e^{-{\log(2p)}}$
\begin{equation*}
\|W\beta_0\|_\infty\leq \lambda_1.
\end{equation*}
Hence 
\begin{equation}
\label{spca.bound.taky2}
\|\beta_0^T W(\beta_0-\hat\beta)\|_\infty \leq \lambda_1\|\beta_0-\hat\beta\|_1.
\end{equation}
%
\vskip 0.2cm
\noindent
By Lemma \ref{cor1} with $\lambda_0 = \sqrt{\frac{\log (2p)}{n}},$ by \eqref{spca.bound.taky2} and using H\"older's inequality,
with probability at least $1-2(J+2)e^{-\log (2p)}$, 
\begin{eqnarray*}
\mathcal E(\betatil) 
&\leq &
\lambda_2'  \|\hat\beta-\beta_0\|_1
+
\lambda_2  \|\hat\beta-\beta_0\|_1
\\
&&\;\;+\;\;4\times 27 \sigma^2\left[3\|\hat\beta-\beta_0\|_1^2{\lambda_0^2}+\sqrt{6} \|\hat\beta-\beta_0\|_1{\lambda_0}\right]\|{{\hat\beta-\beta_0}}\|_2^2 
,
\end{eqnarray*}
where $\lambda_2' = 4\sigma^2(\lambda_0 + \lambda_0^2).$
Next  by the triangle inequality and by the definition of the tuning parameter $\tune$,
$$\lambda_0 \|\beta_0 - \betatil\|_1\leq \lambda_0\|\beta_0\|_1+\lambda_0\tune \leq  C_{\tune}
.$$
Then it follows
$$3\|\hat\beta-\beta_0\|_1^2\lambda_0^2 + \sqrt{6}\|\hat\beta-\beta_0\|_1\lambda_0  \leq
\|\hat\beta-\beta_0\|_1\lambda_0 (3C_T+\sqrt{6}).
$$
But then
\begin{eqnarray*}
\mathcal E(\betatil) 
&\leq &
4 \times 27  \sigma^2 C_T(3C_T +\sqrt{6}) \|{{\hat\beta-\beta_0}}\|_2^2\\
&& \;\;+\;\;
(\lambda_2+\lambda_2')  \|\hat\beta-\beta_0\|_1
.
\end{eqnarray*}
By the condition on the tuning parameter $\lambda$, we have
$\lambda\geq 2(\lambda_2 + \lambda_2'),$ hence
\begin{eqnarray*}
\mathcal E(\betatil) 
&\leq &
4 \times 27  \sigma^2 C_T(3C_T +\sqrt{6}) \|{{\hat\beta-\beta_0}}\|_2^2\\
&&\;\; +\;\;
\lambda/2  \|\hat\beta-\beta_0\|_1
.
\end{eqnarray*}

\noindent
Returning to the oracle inequality, by the condition $\rho-3\eta \geq c_0 \sigma^2 C_T (3C_T+\sqrt{6}),$ where we take
 $c_0:=2\times 4 \times 27$,
we obtain
\begin{equation}
\label{four}
R(\betatil)-R(\beta_0)  
+ \lambda \|\betatil\|_1 \leq \lambda\|\beta_0\|_1+ (\rho-3\eta)/2 \|\beta_0 - \betatil\|_2^2 + \lambda/2 \|\beta_0 - \betatil\|_1.
\end{equation}
The oracle inequalities then follow by the usual techniques (see e.g. \cite{hds}), since the population risk 
satisfies $R(\betatil)-R(\beta_0) \geq (3\rho-\eta) \|\betatil-\beta_0\|_2^2$ as already derived above.

\end{proof}

\subsubsection{Oracle inequalities for nodewise regression}
\label{sec:proof.nodewise}

In this section, we derive the rates of convergence for the estimator $\hat\Theta$ defined in Algorithm \ref{alg:nodewisea}. 
These results are contained in Lemmas \ref{gamma} and \ref{tau.spca} below. 
\noindent
Recall the definition of  the population parameters $\gamma_j^0$ from \eqref{gamma0} and define 
\begin{equation}\label{tauj0}
\tau_j^2 = \mathbb E \ddot R_n(\beta_0)_{j,j}-\ddot R_n(\beta_0)_{j,-j}\gamma^0_{j}.
\end{equation}
We now summarize several relationships that will be used throughout the proofs without further reference.
One can easily check that the definition of $\gamma_j^0$  implies 
$$\gamma_j^0 = (\ddot R(\beta_0)_{-j,-j})^{-1}\ddot R(\beta_0)_{j,-j},$$
provided that  the matrix is $\ddot R(\beta_0)_{-j,-j}$ is invertible.
One can also verify that $\tau_j^2$
satisfies
\begin{eqnarray*}
\frac{1}{\tau_j^2} =
{\Theta^0_{jj}}.
\end{eqnarray*}
It is moreover not difficult to calculate the following relations, which will be used throughout the proofs
\begin{eqnarray*}
\Lambda_{\min}(\ddot R(\beta_0)) &=& \phi^2_{\max} -\phi_2^2 ,\\
\Lambda_{\max}(\ddot R(\beta_0)) &=& 2\phi^2_{\max}, \\ 
\Lambda_{\min}(\Theta_0)         &=& 1/(2\phi^2_{\max}),\\
\Lambda_{\max}(\Theta_0)         &=& 1/( \phi^2_{\max} - \phi_2^2 ). 
\end{eqnarray*}
This also implies ${1}/{\tau_j^2}\leq \Lambda_{\max}(\Theta_0)  \leq 1/( \phi^2_{\max} - \phi_2^2 )$
and hence $\tau_j^2\geq  \phi^2_{\max} - \phi_2^2 .$
To simplify notation, in this section we denote $\alpha:=\phi^2_{\max} - \phi_2^2.$

\begin{lemma}
\label{gamma}
Assume Condition \ref{design.spca} with parameter $\sigma$,
let 
$$\lambda_0=\sqrt{\log (2p)/n},$$
$$\bar\lambda_1 \geq 16\sigma^2(\|\gamma_j^0\|_2+1)[\lambda_0+\lambda_0^2]$$
$$ \rho-3\eta \geq c_0 \sigma^2 C_{\tune}[3C_{\tune}+\sqrt{6}],$$
where $c_0$ is a suitable universal constant.
Let the tuning parameters $\lambda_j,\tunej,$ $j=1,\dots,p$ of the program \eqref{est.fast} satisfy
\begin{align}
\label{tuning.par1234} 
\lambda_j\geq 2\bar\lambda_1,
\end{align}
$\tunej \leq C_{\tune} /(2\lambda_0)$ and $\|\gamma_j^0\|_1\leq \tunej.$
Then any stationary point $\gamatil$ as defined in \eqref{stationary} satisfies with probability at least 
$1-2(J+1)pe^{-2{\log(2p)}},$ where $J=\lceil \log T\rceil$, 
\begin{align}
\label{EqnStatBounds}
\max_{j=1,\dots,p}
\|\gamatil - \gamma_j^0\|_2^2 +\lambda_j\|\gamatil - \gamma_j^0\|_1 \le \max_{j=1,\dots,p} \frac{C_1 s_j\lambda_j^2}{(\rho-3\eta)^2},
\end{align}
where $s_j = \|\gamma_j^0\|_0$ and $C_1$ is a universal constant.

\end{lemma}

\begin{proof}[Proof of Lemma \ref{gamma}]
The proof is similar to the proof of Lemma \ref{oracle}. For simplicity, we denote the loss function by 
$L_n(\gamma_j):= \Gamma_j^T \ddot R_n(\hat\beta) \Gamma_j$ where the notation is as in \eqref{lasso-type}. Define $L(\gamma_j):= \Gamma_j^T \ddot R(\hat\beta) \Gamma_j.$ The derivatives are denoted by dots.
The definition of the stationary point $\gamatil$ implies
\begin{equation}
\label{kkt123}
(\dot L_n(\gamatil) + \lambda \hat Z)^T (\gamma_j^0-\gamatil) \geq 0,
\end{equation}
where $\hat Z$ is the sub-differential of the $\ell_1$-norm evaluated at $\gamatil.$
By Taylor expansion of the population loss, we obtain
\begin{equation}
\label{taylor123}
L(\gamma_j^0) - L(\gamatil) = \dot L(\gamatil)^T(\gamma_j^0-\gamatil) + \frac{1}{2}(\gamma_j^0-\gamatil)^T\ddot R({\hat \beta}) (\gamma_j^0-\gamatil).
\end{equation}
We have  $
 \Lambda_{\min}(\ddot R({\hat \beta}))\geq 2(\rho-3\eta)$ by Lemma \ref{eigen.spca}.
Thus, combining \eqref{kkt123} and \eqref{taylor123} and rearranging yields
\begin{equation}
\label{four}
L(\gamatil) -L(\gamma_j^0) 
+ (\dot L_n(\gamatil ) - \dot L(\gamatil))^T(\gamma_j^0-\gamatil) -
 \lambda_j \hat Z^T (\gamma_j^0-\gamatil) \leq 0.
\end{equation}
Then it follows (using that $\hat Z^T \gamatil = \|\gamatil\|_1$ and $|\hat Z^T \gamma_j^0|\leq \|\hat Z\|_\infty\|\gamma_j^0\|_1\leq \|\gamma_j^0\|_1$)
\begin{equation}
\label{four}
L(\gamatil)-L(\gamma_j^0)  
+ \lambda_j \|\gamatil\|_1 \leq \lambda_j\|\gamma_j^0\|_1+ |(\dot L_n(\gamatil) - \dot L(\gamatil))^T(\gamma_j^0-\gamatil)|.
\end{equation}
It remains to bound the term 
\begin{eqnarray*}
(\dot L_n(\gamatil) - \dot L(\gamatil))^T(\gamma_j^0-\gamatil)
&=&
2(\hat\Sigma_{j,-j}-\Sigma_{j,-j}^0)^T(\gamma_j^0-\gamatil) \\
&& + \;2 \gamatil^T(\hat\Sigma_{-j,-j}-\Sigma_{-j,-j}^0)(\gamma_j^0-\gamatil)
.
\end{eqnarray*}
We may use the same bounds as in Lemma \ref{oracle}, only now we need to consider maximum over all $j=1,\dots,p$. Hence by union bound, we  obtain
with probability at least
$1-2(J+1)pe^{-\log(2p)}$,
with $\lambda_0 = \sqrt{\frac{2\log (2p)}{n}},$ and 
$$\bar\lambda_1\geq 16\sigma^2(\|\gamma_j^0\|_2+1)(\lambda_0 + \lambda_0^2),$$
and by the definition of tuning parameters $T_j$,
\begin{eqnarray*}
|(\dot L_n(\gamatil) - \dot L(\gamatil))^T(\gamma_j^0-\gamatil)| 
&\leq &
4 \times 27 \sigma^2 C_T (3C_T+\sqrt{6}) \|{{\gamatil-\gamma_j^0}}\|_2^2\\
&&\;\; +\;\;
\lambda/2  \|\hat\beta-\beta_0\|_1
.
\end{eqnarray*}

\noindent
Returning to the oracle inequality, we have
\begin{equation}
\label{four}
L(\gamatil)-L(\gamma_j^0)  
+ \lambda_j \|\gamatil\|_1 \leq \lambda_j\|\gamma_j^0\|_1+(\rho-3\eta)/2\|\gamma_j^0-\gamatil\|_2^2+ \lambda_j/2\|\gamma_j^0-\gamatil\|_1.
\end{equation}
By the usual techniques (see e.g. \cite{hds}), we obtain the oracle inequalities.

\end{proof}

\begin{lemma}\label{tau.spca}
Suppose that conditions of Lemma \ref{gamma} are satisfied 
and denote 
$ \mu:= \|\hat\beta-\beta_0\|_2$
and $\hat c_j:= \hat\gamma_j - \gamma_j^0.$
Then for $\lambda_0 \geq \|\hat\Sigma-\Sigma_0\|_\infty,$ it holds
$$|\hat\tau_{j}^2-\tau_{j}^2|\leq r_{\tau,j},$$
where
\begin{eqnarray*}
r_{\tau,j} &:=&
\lambda_0 (\deltagammaone + \sqrt{s_j+1}\phi_{\max}^2/\egap)\\[3pt]
&&\;\;+\;\;
2\phi_{\max}^2\deltagammatwo\\[3pt]
&&\;\;+\;\;
2{{\mu}}({{\mu}} + 2\phi_{\max}) [\deltagammatwo + \phi_{\max}^2/\egap].
\end{eqnarray*}

\noindent
Moreover, if $\alpha - r_{\tau,j}>0,$
\begin{eqnarray*}
\|\hat\Theta_{j}-\Theta_{j}^0\|_1 &\leq&
\frac{1}{\egap} \deltagammaone + \left(1+\sqrt{s_j}\frac{\phi_{\max}^2}{\egap}\right) \frac{r_{\tau,j}}{\egap - r_{\tau,j}},
\\
%
\|\hat\Theta_{j}-\Theta_{j}^0\|_2 &\leq&
\frac{1}{\egap} \deltagammatwo + \left(1+\frac{\phi_{\max}^2}{\egap}\right) \frac{r_{\tau,j}}{\egap - r_{\tau,j}},
\end{eqnarray*}
where $\alpha=\Lambda_1-\Lambda_2.$

\end{lemma}

\begin{proof}[Proof of Lemma \ref{tau.spca}]

First, one can easily show from the KKT condi\-tions for the nodewise Lasso that
 $\hat\tau_j^2 = \ddot R_n(\hat\beta)_j^T \hat\Gamma_j.$
Consider the decomposition
\begin{eqnarray*}
\hat\tau_j^2-\tau_j^2 &=& 
\underbrace{(\ddot R_n(\hat\beta)_{j} 
-
 \ddot R(\hat\beta)_{j})^T\hat\Gamma_{j}}_{i}
\\
&&+
 \underbrace{(\ddot R(\hat \beta)_{j}
-
\ddot R(\beta_0)_{j})^T\hat\Gamma_{j}}_{ii}
\\
&&+
 \underbrace{\ddot R(\beta_0)_{j}^T(\hat\Gamma_{j}
-\Gamma^0_{j})}_{iii}.
\end{eqnarray*}
We need to bound the terms $i,ii,iii.$ Before doing so, we prepare a few preliminary results.
Firstly, 
$$\|\Gamma^0_{j}\|_2 = ((\Theta^0_j)^T (\Theta^0_j))^{1/2} /\Theta^0_{jj} \leq \Lambda_{\max}(\Theta_0)/ \Lambda_{\min}(\Theta_0)
\leq 2\phi_{\max}^2/\egap.$$
Next observe, 
\begin{eqnarray}
\nonumber
|\|\hat\beta\|_2^2-\|\beta_0\|_2^2 | &\leq & |(\hat\beta-\beta_0)^T (\hat\beta-\beta_0)| + 2 |\beta_0^T (\hat\beta-\beta_0)|
\\\nonumber
&\leq &
\|\hat\beta-\beta_0\|_2^2 + 2\|\beta_0\|_2\|\hat\beta-\beta_0\|_2
\\\label{spca.bound.norms}
&\leq &
{{\mu}}^2 + 2\phi_{\max} {{\mu}}.
\end{eqnarray}
Hence,
\begin{eqnarray*}
\|(\ddot R(\hat\beta) - \ddot R(\beta_0))e_j\|_2
&=&
\|(\|\hat\beta\|_2^2-\|\beta_0\|_2^2)e_j + \hat\beta_j \hat\beta - \beta^0_j \beta_0\|_2
\\
&\leq &
|\|\hat\beta\|_2^2-\|\beta_0\|_2^2| + |\hat\beta_j|\|\hat\beta-\beta_0\|_2 \\
&& \;\;+\;\; \|\beta_0\|_2 |\hat\beta_j-\beta_j^0|
\\
&\leq&
|\|\hat\beta\|_2^2-\|\beta_0\|_2^2|  + \|\hat\beta\|_2 \|\hat\beta-\beta_0\|_2 \\
&&\;\;+\;\; \|\beta_0\|_2 \|\hat\beta-\beta_0\|_2 
\\
&\leq &
2{{\mu}}^2 + 2\phi_{\max} {{\mu}}.
\end{eqnarray*}
Now using the above preliminaries, we obtain the bounds for $i,ii,iii.$ Firstly, observing that 
$\|\Gamma_j^0\|_1\leq \sqrt{\spag+1}\|\Gamma_j^0\|_2$,
\begin{eqnarray*}
|i| 
&\leq &
\|\ddot R_n(\hat\beta)_{j} - \ddot R(\hat\beta)_{j}\|_\infty\|\hat\Gamma_j\|_1
\\
&\leq &
\lambda_0 (\deltagammaone + 2\sqrt{\spag+1}\phi_{\max}^2/\egap),
\end{eqnarray*}
where $\|\hat\Sigma-\Sigma_0\|_\infty \leq \lambda_0.$ 
Moreover,
\begin{eqnarray*}
|ii| &\leq & \|\ddot R(\hat\beta)_{j} - \ddot R(\beta_0)_{j}\|_2\|\hat\Gamma_j\|_2
\\
&\leq &
2{{\mu}}({{\mu}} + \phi_{\max}) (\deltagammatwo + 2\phi_{\max}^2/\egap)
.
\end{eqnarray*}
Next
\begin{eqnarray*}
|iii| 
&\leq &
\|\ddot R(\beta_0)_{j}\|_2 \deltagammatwo
\\
&=& 2\phi_{\max}^2\deltagammatwo. 
\end{eqnarray*}
Thus collecting the results above,
$$|\hat\tau_j^2-\tau_j^2  |\leq 
r_{\tau,j},$$
where 
\begin{eqnarray*}
r_{\tau,j} &:=&
\lambda_0 (\deltagammaone + \sqrt{s_j+1}\phi_{\max}^2/\egap)\\[3pt]
&&\;\;+\;\;
2\phi_{\max}^2\deltagammatwo\\[3pt]
&&\;\;+\;\;
2{{\mu}}({{\mu}} + 2\phi_{\max}) [\deltagammatwo + \phi_{\max}^2/\egap].
\end{eqnarray*}
By the mean value theorem,
$$|\frac{1}{\hat\tau_j^2} - \frac{1}{\tau_j^2}| \leq \frac{1}{\tilde\tau^2}|\hat\tau_j^2-\tau_j^2 |,$$
for some intermediate point $\tilde\tau^2.$
But we have
$$\tilde\tau^2 \geq \tau_j^2 -|\tilde\tau_j^2-\tau_j^2 | \geq 
 1/\Theta^0_{jj}- |\hat\tau_j^2-\tau_j^2 |\geq\egap -r_{\tau,j},$$
Hence, assuming that $\egap -r_{\tau,j}>0,$
$$|\frac{1}{\hat\tau_j^2} - \frac{1}{\tau_j^2}| \leq  \frac{r_{\tau,j}}{\egap - r_{\tau,j}}.$$
Then we can easily obtain the rates of convergence for $\hat\Theta_j$ using the bound
\begin{eqnarray*}
\|\hat\Theta_{j}-\Theta_{j}^0\|_1 &=&
\|\hat\Gamma_{j}/\hat\tau_j^2 - \Gamma_{j}/\tau_j^2\|_1
\\
&\leq & 
\deltagammaone /\hat\tau_j^2
 + 
\|\Gamma^0_j\|_1 |1/\hat\tau_j^2 - 1/\tau_j^2|.
\end{eqnarray*}
Hence
\begin{eqnarray*}
\|\hat\Theta_{j}-\Theta_{j}^0\|_1 &\leq&
\frac{1}{\egap} \deltagammaone+ \left(1+2\sqrt{s_j+1}\frac{\phi_{\max}^2}{\egap}\right) \frac{r_{\tau,j}}{\egap - r_{\tau,j}}
\end{eqnarray*}
Similarly follow the rates for $\|\hat\Theta_j-\Theta_j^0\|_2.$
\end{proof}

\begin{proof}[Proof of Lemma \ref{nodewise.asymp.spca}]
Follows from Lemmas \ref{gamma} and \ref{tau.spca} by noting that $\alpha =\phi_1^2 - \phi_2^2 \geq \rho^2\geq c$ for a universal constant,
and $\phi_{\max}\leq C_{\max}$. Then by Lemma \ref{gamma} it follows 
$\max_{j=1,\dots,p}\|\hat\gamma_j-\gamma_j^0\|_2^2 = \mathcal O_P(\max_{j}s_j\lambda_j^2)$ and 
$\max_{j=1,\dots,p}\|\hat\gamma_j-\gamma_j^0\|_1 = \mathcal O_P(\max_{j}s_j\lambda_j).$
By Lemma \ref{tau.spca} (since $\|\hat\Sigma-\Sigma_0\|_\infty = \mathcal O_P(\sqrt{\log p/n})$), it then follows
 $$\max_{j=1,\dots,p}\|\hat\Theta_j-\Theta_j^0\|_2^2 = \mathcal O_P(\max_{j}s_j\lambda_j^2),$$ 
and 
 $$\max_{j=1,\dots,p}\|\hat\Theta_j-\Theta_j^0\|_1= \mathcal O_P(\max_{j}s_j\lambda_j).$$

\end{proof}

\subsubsection{Asymptotic normality}
\label{sec:proof.normal}

\begin{proof}[Proof of Theorem \ref{ci}] 
Using Taylor expansion of the function $\beta \mapsto \hat\Theta_j^T\dot R_n(\beta)$ around $\beta_0$ we obtain:
$$\hat\Theta_j^T\dot R_n(\hat\beta) = \hat\Theta_j^T\dot R_n(\beta_0) + \hat\Theta_j^T \ddot R_n(\tilde\beta) (\hat\beta-\beta_0),$$
where $\tilde\beta = \alpha \hat\beta + (1-\alpha ) \beta_0$ for some $\alpha\in[0,1].$
Then for the de-sparsified estimator, we may write the decomposition
\begin{eqnarray*}
\hat b_j - \beta_j^0 &=&
\hat \beta_{j} -\beta^0_{j} -
\hat\Theta_{j}^T \dot{R}_n(\hat \beta)\\
&=&
-\;(\Theta_{j}^0)^T\dot R_n( \beta_0)
\\
&&\; - \underbrace{(\hat\Theta_{j} - \Theta^0_{j})^T \dot{R}_n( \beta_0)}_{i}
\\
&& \;+ \underbrace{ \hat \beta_j -\beta^0_j - \hat\Theta_{j}^T\ddot R_n(\hat\beta)(\hat\beta-\beta_0)}_{ii }
\\
&&
\;-\underbrace{\hat\Theta_{j}^T(\ddot R_n(\tilde\beta) - \ddot R_n(\hat\beta)) (\hat\beta-\beta_0)}_{iii},
\end{eqnarray*}
We first bound $ii$ using H\"older's inequality and the KKT conditions for nodewise Lasso for inversion of $\ddot R_n(\hat\beta)$.
The estimator $\hat\gamma_j$ is defined as any stationary point of the program \eqref{lasso-type}, but as we have shown oracle inequalities 
for $\hat\gamma_j$, for $n$ sufficiently large, $\hat\gamma_j$ must lie in the interior of the feasible set and hence the KKT conditions
$-2\ddot R_n(\hat\beta)_{j,-j} + 2\ddot R_n(\hat\beta)_{-j,-j}\hat\gamma_j +\lambda_j\partial\|\hat\gamma_j\|_1=0,$ are satisfied with high probability.
The KKT conditions for nodewise regression imply that $\|\ddot R_n(\hat\beta)\hat\Theta_j-e_j\|_\infty=\mathcal O( \lambda_j/\hat\tau_j^2)$ (see e.g. \cite{vdgeer13}).   Hence
\begin{eqnarray*}
|ii|
& =& 
\|(\hat \beta -\beta_0)(e_j - \hat\Theta_{j}^T\ddot R_n(\hat\beta))\|_\infty 
\leq 
\|\hat \beta -\beta_0\|_1 \|e_j - \hat\Theta_{j}^T\ddot R_n(\hat\beta)\|_\infty 
\\
&\leq &
\|\hat \beta -\beta_0\|_1 \lambda_j/\hat\tau_j^2 
.
\end{eqnarray*}
Next we bound $iii$  using the Cauchy-Schwarz inequality 
$$|iii| \leq \|\hat\beta -\beta_0 \|_2 \|\hat\Theta_j^T (\ddot R_n(\tilde\beta) - \ddot R_n(\hat\beta))\|_2.$$
By the definition of $\tilde\beta$ it follows that $\|\tilde\beta-\beta_0\|_2 \leq \|\hat\beta-\beta_0\|_2
$.
But then
\begin{eqnarray*}
&& \|\hat\Theta_j^T (\ddot R_n(\tilde\beta) - \ddot R_n(\hat\beta))\|_2
\\[0.1cm]
&&\leq 
\|\hat\Theta_j^T (\ddot R_n(\tilde\beta) - \ddot R_n(\beta_0))\|_2
+ \|\hat\Theta_j^T (\ddot R_n(\hat\beta) - \ddot R_n(\beta_0))\|_2\\
&&\leq 
\|\hat\Theta_j^T (\|\tilde\beta\|_2^2 - \|\beta_0\|_2^2)\|_2 + \|\hat\Theta_j^T \tilde\beta\tilde\beta^T -\hat\Theta_j^T \beta_0\beta_0^T \|_2
 \\
&&\leq 
\|\hat\Theta_j\|_2 |\|\tilde\beta\|_2^2 - \|\beta_0\|_2^2|
 + |\hat\Theta_j^T \tilde\beta|\|\tilde\beta-\beta_0\|_2 + \|\hat\Theta_j^T\|_2\|\tilde\beta-\beta_0\|_2\|\beta_0 \|_2
 \\
&&\leq 
\|\hat\Theta_j\|_2 \left(|\|\tilde\beta\|_2^2 - \|\beta_0\|_2^2|
 +  \|\tilde\beta\|_2 \|\tilde\beta-\beta_0\|_2 + \|\tilde\beta-\beta_0\|_2\|\beta_0 \|_2 \right)
\\
&&\leq 
\|\hat\Theta_j\|_2(2{{\mu}}^2 + 2\phi_{\max}{{\mu}})\\
&&\leq 
(\|\hat\Theta_j-\Theta_j^0\|_2 +1/\egap )(2\mu_{\hat\beta}^2 + 2\phi_{\max}{{\mu}}).
\end{eqnarray*}
where we used the bound \eqref{spca.bound.norms} from the proof of Lemma \ref{tau.spca}.
Therefore,
$$|iii|\leq (\|\hat\Theta_j-\Theta_j^0\|_2 +1/\egap )(2{{\mu}}^2 + 2\phi_{\max}{{\mu}}){{\mu}}.$$
Finally,
\begin{eqnarray*}
|i| &\leq & \|\hat\Theta_{j} - \Theta^0_{j}\|_1 \| \dot{R}_n( \beta_0)\|_\infty\\
& =&
\|\hat\Theta_{j} - \Theta^0_{j}\|_1 \| \dot{R}_n( \beta_0)-\dot R(\beta_0)\|_\infty\\
&=&
\|\hat\Theta_{j} - \Theta^0_{j}\|_1 \| (\hat\Sigma-\Sigma_0)\beta_0\|_\infty
\\
&\leq & \lambda_0\|\hat\Theta_{j} - \Theta^0_{j}\|_1,
\end{eqnarray*}
for $\lambda_0 \geq \|(\hat\Sigma-\Sigma_0)\beta_0\|_\infty.$
Hence the bound for the remainder is
\begin{eqnarray*}
&&\max_{j=1,\dots,p}|\text{rem}_j|  \\
&& \quad := \max_{j=1,\dots,p} |i|+|ii| + |iii| \\
&&\quad\leq  \max_{j=1,\dots,p}\|\hat \beta -\beta_0\|_1 \lambda_j/\hat\tau_j^2  \\
&&\quad\quad+ \;(\|\hat\Theta_j-\Theta_j^0\|_2 +1/\egap )(2\|\hat\beta-\beta_0\|_2^2 +  2\phi_{\max}\|\hat\beta-\beta_0\|_2)\|\hat\beta-\beta_0\|_2
\\
&&\quad\quad  + \;\|\hat\Theta_j - \Theta_j^0\|_1 \lambda_0.
\end{eqnarray*}
We now combine the last bound with the result of Lemma \ref{tau.spca} and probability results from Section \ref{sec:spca.prob}.
In particular, under the Condition \ref{design.spca}  and the assumptions $\phi_{\max}\leq C_{\max}$, $\rho-3\eta \geq c>0$ and 
the assumed sparsity conditions, 
 it follows
that
\begin{eqnarray*}
\max_{j=1,\dots,p}|\text{rem}_j| &=&
 \mathcal O_P\left(
\max_{j=1,\dots,p} \max(s,s_j)\max\left(\lambda^2, \lambda_j^2,
\frac{\log(2p)}{n}\right)
\right) \\
&=& o_P\left(\frac{1}{\sqrt{n}}\right).
\end{eqnarray*}
Thus we conclude that
\begin{eqnarray*}
\hat \beta_{} -\beta_{} -
\hat\Theta_{} \dot{R}_n(\hat \beta)
=
-\Theta_{0} \dot{R}_n( \beta_0)+o_P(1/\sqrt{n}).
\end{eqnarray*}
Finally, one can easily check that the random variable 
  $(\Theta^0_{j})^T \Xobsi (\Xobsi)^T\beta_0$ has bounded fourth moments  
	under Condition \ref{design.spca} with $\sigma$ is a universal constant and if $\phi_{\max}\leq C_{\max}, \rho\geq c>0$. Hence we may use the Lindeberg central limit theorem on the term
$(\Theta^0_{j})^T \dot R_n(\beta_0)$
to obtain that (assuming $1/\sigma_j= \mathcal O(1)$)
$$\sqrt{n}(\Theta^0_{j})^T \dot R_n(\beta_0)/\sigma_j \rightsquigarrow \mathcal N(0,1).$$
This then implies 
$$\sqrt{n}(\hat b_{j} -\beta^0_{j})/\sigma_j=\sqrt{n}(\Theta^0_{j})^T \dot R_n(\beta_0)/\sigma_j + o_P\left(\frac{1}{\sigma_j }\right) \rightsquigarrow \mathcal N(0,1).$$
\end{proof}

\begin{proof}[Proof of Theorem \ref{eigenvalue}]
By Theorem \ref{ci}, we have the asymptotic expansion
$$\hat b_{} -\beta_{0} = -\Theta_0 \dot R_n(\beta_0 )+ \text{rem},$$
with $\|\text{rem}\|_\infty =\mathcal O_P(s
\max(\lambda^2,\max_{j=1,\dots,p} \lambda_j^2,
\log p/n)
)
.$
Hence
\begin{eqnarray*}
\|\hat \beta\|_2^2 - \|\beta_0\|_2^2  
&=&
2\beta_0^T (\hat \beta-\beta_0) + (\hat \beta- \beta_0)^T(\hat \beta-\beta_0)
\\
&=&
2\beta_0^T (\hat \beta-\beta_0-\hat\Theta^T \dot R_n(\hat\beta)) + 2\beta_0^T \hat\Theta^T \dot R_n(\hat\beta)+ \|\hat \beta-\beta_0\|_2^2
\\
&=&
2\beta_0^T (\hat b-\beta_0) + \|\hat \beta-\beta_0\|_2^2
\\
&=&
-2\beta_0^T \Theta_0\dot R_n(\beta_0) + 2\hat\beta^T \hat\Theta^T \dot R_n(\hat\beta) \\
&& +\; \text{rem}_2
,\end{eqnarray*}
where the remainder  $\text{rem}_2$ can be bounded
\begin{eqnarray*}
|\text{rem}_2| 
&:=&
{2\beta_0^T \text{rem}
+ 2(\beta_0-\hat\beta)^T \hat\Theta^T \dot R_n(\hat\beta) +\|\hat \beta-\beta_0\|_2^2}
\\
&\leq &
2\|\beta_0\|_1 \|\text{rem}\|_\infty
+ 2\|\beta_0-\hat\beta\|_1 \vertiii{\hat\Theta^T}_1 \| \dot R_n(\hat\beta)\|_\infty + \|\hat \beta-\beta_0\|_2^2
\\
&\leq &
\mathcal O_P(\sqrt{s}\max(s,s_j)^{}\max(\lambda^2,\max_{j=1,\dots,p} \lambda_j^2,\log p/n)) 
\\&&\;
+\;\;\mathcal O_P(\sqrt{s_j}\max(s,s_j)^{}\lambda^2) + \mathcal O_P(s\lambda^2)\\
&=&
\mathcal O_P(\max(s,s_j)^{3/2}\max(\lambda^2,\max_{j=1,\dots,p} \lambda_j^2,\log p/n))
. 
\end{eqnarray*}
Hence, under the sparsity conditions, 
 we obtain
$$\|\hat \beta\|_2^2 - \|\beta_0\|_2^2 - 2\hat\beta^T \hat\Theta^T \dot R_n(\hat\beta) =
-2\beta_0^T\Theta_0 \dot R_n(\beta_0) +
o_P(1/\sqrt{n}).$$
As in the proof of Theorem \ref{ci}, it follows that the zero-mean random variable $2\beta_0^T\Theta_0 \dot R_n(\beta_0) $ has bounded fourth moments.
Asymptotic normality then follows  by  an application of the Lindeberg central limit theorem.

\end{proof}

\begin{lemma}
\label{variance}
If $X_i\sim \mathcal N(0,\Sigma_0),i=1,\dots,n$ and $\rho>0$, then it holds that
$$\sigma_j^2:=
\emph{var}((\Theta^0_j)^T \hat\Sigma \beta_0) = \frac{\beta^0_j}{2} + \|\beta_0\|_2^4 \sum_{i=1,i\not = j}^p u_{i,j}^2 \frac{\Lambda_{i}}{(\Lambda_i - \Lambda_{\max})^2}     
,$$
and
$$\sigma^2 := 4\emph{var}(\beta_0^T\Theta_0 \hat\Sigma \beta_0)=2\Lambda_{\max}.
$$
 
\end{lemma}
\begin{proof}[Proof of Lemma \ref{variance}]
First under normality, it is well-known that
$$\sigma_j^2:=
(\Theta^0_{j})^T \Sigma_0 \Theta^0_j \|\beta_0\|_2^4 + [\|\beta_0\|_2^2(\Theta^0_j)^T \beta_0]^2,$$
and 
$$\sigma_{\Lambda}^2 := 
 4\beta_0^T\Theta_0 \Sigma_0 \Theta_0\beta_0 \|\beta_0\|_2^4  + 4(\|\beta_0\|_2^2 \beta_0^T\Theta_0 \beta_0)^2. 
$$
\vskip 0.1cm
\noindent
We first calculate $\Theta_0:=\ddot R(\beta_0)^{-1}.$
We write the eigendecomposition of $\Sigma_0$ as $\Sigma_0 = U^T \Lambda U$, for some $U$ such that $U^TU = 1$. By the condition $\rho>0$, the first column of $U$ is $u_1 = \beta_0/\|\beta_0\|_2$.
Then we may write 
\begin{eqnarray*}
\ddot R(\beta_0) &=&
 U^T (-\Lambda + \|\beta_0\|_2^2 + 2U\beta_0\beta_0^T U^T)U \\
&=&
 U^T (-\Lambda + \|\beta_0\|_2^2 + 2\|\beta_0\|_2^2 e_1e_1^T  )U
,
\end{eqnarray*}
which can be easily inverted
$$\Theta_0= \ddot R(\beta_0)^{-1} = 
 U^T D U
,$$
where 
$$D:=\text{diag}\left(
\frac{1}{2\|\beta_0\|_2^2} ,\frac{1}{\|\beta_0\|_2^2 - \Lambda_2(\Sigma_0)}, 
\dots, 
\frac{1}{\|\beta_0\|_2^2 - \Lambda_p(\Sigma_0)}\right).$$
Then we have 
$$\Theta_0 \beta_0 = U^T D U\beta_0 = U^T \frac{1}{2\|\beta_0\|_2^2} \|\beta_0\|_2e_1 = \beta_0/(2\|\beta_0\|_2^2),$$
and
$$\Theta_0 ^T \Sigma_0 \Theta_0 = U^T D\Lambda D U.$$
Finally, we conclude
\begin{eqnarray*}
\sigma_{\Lambda}^2 &=&4\beta_0^T \Theta_0^T \Sigma_0 \Theta_0 \beta_0  + 4(\beta_0^T \Theta_0 \Sigma_0 \beta_0 )^2\\
&
=&
4 \frac{1}{4} \|\beta_0\|_2^4 + 4 \frac{1}{4} \|\beta_0\|_2^4 =2\|\beta_0\|_2^4  = 2\Lambda_{\max}, 
\end{eqnarray*}
and
\begin{eqnarray*}
\sigma_j^2&:=&
(\Theta^0_{j})^T \Sigma_0 \Theta^0_j \beta_0^T \Sigma_0 \beta_0 + ((\Theta^0_j)^T \Sigma_0\beta_0)^2
\\[0.1cm]
&=&
(\Theta^0_{j})^T \Sigma_0 \Theta^0_j \beta_0^T \Sigma_0 \beta_0 + ((\Theta^0_j)^T \Sigma_0\beta_0)^2
\\
&=&
\|\beta_0\|_2^4  U^T \Lambda^{1/2} D^{2} \Lambda^{1/2}U+ \frac{\beta^0_j}{4}
\\
&=& \frac{\beta^0_j}{4} +\frac{\beta^0_j}{4} +\|\beta_0\|_2^4 \sum_{i=1,i\not = j}^p u_{i;j}^2 \frac{\Lambda_{i}}{(\Lambda_i - \Lambda_{\max})^2}  .
\end{eqnarray*}
\end{proof}

\ifsubsection
\section{Probabilistic bounds for the empirical process }
\else
\subsection{Probabilistic bounds for the empirical process }
\fi
\label{sec:spca.prob}
We collect probabilistic results needed to bound the empirical process part related to the estimators $\hat\beta,\hat\gamma_j.$
Recall the definition of a sub-Gaussian matrix from \condi \ref{design.spca}.

\begin{lemma}
\label{sigma.con}
If $X \in \mathbb R^{n\times p}$ is a sub-Gaussian matrix with parameter $\sigma$, then 
for any \emph{fixed} vector $\beta$, with probability at least 
$1-2e^{-\log(2p)}$
it holds
$$\|(\hat\Sigma-\Sigma_0)\beta\|_\infty \leq 4\|\beta\|_2 \sigma^2 \left(\sqrt{\frac{2\log(2p)}{n}}+\frac{2\log(2p)}{ n}\right).$$
\end{lemma}
\begin{proof}[Proof of Lemma \ref{sigma.con}]
The result follows from Lemma 14.13 in \cite{hds}. 
\end{proof}

\begin{lemma}
\label{dev}
If $X \in \mathbb R^{n\times p}$ is a sub-Gaussian matrix with parameter $\sigma$, then for all $t>0$
\begin{eqnarray*}
&& P\left(
\sup_{\theta\in \mathbb R^p: \|\theta\|_2 =1, \|\theta\|_0\leq M}
|\frac{\|X\theta\|_2^2}{n} -  \frac{\mathbb E\|X\theta\|_2^2}{n} | \;\geq \; 2\sigma^2(t+\sqrt{2t})
 \right)
\\[0.3cm]
&&\quad\quad \leq
2\exp\left(-nt+ 2M{{\log (2p)}}\right)
\end{eqnarray*}

\end{lemma}

\begin{proof}[Proof of Lemma \ref{dev}]
This lemma is essentially Lemma 15 in \cite{LohWai12}, but we  apply a  slightly different version of Berstein's inequality, namely
Lemma 14.9 in \cite{hds}. 
\end{proof}

\noindent
Denote $\mathbb B_r(M)=\{\theta\in\mathbb R^p: \|\theta\|_r\leq M\}$ for $r \geq 0 $.

\begin{lemma}[Lemma 11 in \cite{LohWai12}]
\label{l1tol0}
For any constant $s\geq 1$, it holds
$$\mathbb B_1(\sqrt{s})\cap \mathbb B_2(1) \subseteq 3\emph{cl}(\emph{conv}(\mathbb B_0(s)\cap \mathbb B_2(1))),$$
where $\emph{cl}(\cdot)$ denotes the topological closure of a set and  $\emph{conv}(\cdot)$ denotes the convex hull.
\end{lemma}


\begin{lemma}
\label{cor1}
Suppose that $X \in \mathbb R^{n\times p}$ is a sub-Gaussian matrix with parameter $\sigma$.
Let $J:= \lceil \log_2(T) \rceil $, 
and
$$\lambda_0 = \sqrt{\frac{2\log (2p)}{n}},$$ 
$$\lambda_1 = 4\sigma^2(\lambda_0 + \lambda_0^2).$$ 
Then with probability at least
$1-2(J+1)2e^{-{\log(2p)}},$
it holds
$$\forall \;\theta\in \mathcal B, \|\theta\|_1\leq T:$$
\begin{eqnarray*}
&&|\theta^TW \theta| 
\leq  
\lambda_1 \|\theta\|_1 + 
\delta_{\|\theta\|_1} \|\theta\|_2^2,
\end{eqnarray*}
where
$$\delta_M:=
 4\times 27 \sigma^2
\left[3M^2{\lambda_0^2}+\sqrt{6} M{\lambda_0}\right].$$

\end{lemma}
\begin{proof}
\noindent
Consider the set
$$A:=\{\theta\in\mathcal B: \|\theta\|_1\leq T\},$$
and the decomposition 
$$A= A_0 \cup A_0^c,$$
where
$$A_0:=\{\theta\in A:\|\theta\|_1\leq 1\}.$$
We denote $W:=\hat\Sigma-\Sigma_0.$
First note that for $\lambda_1 = 4\sigma^2(\lambda_0 + \lambda_0^2),$ by Lemma \ref{sigma.con} it follows that
with probability at least $1-\alpha_1$, where $\alpha_1:=2e^{-{\log(2p)}}$
\begin{equation}
\label{spca.bound.taky}
\|W\|_\infty \leq \lambda_1.
\end{equation}
If we are on the set $A_0$ and then by H\"older's inequality and bound \eqref{spca.bound.taky}, with probability at least $1-\alpha_1$,
for all $\theta\in A$
\begin{eqnarray}
\label{A0}
|\theta^T W\theta|
&\leq &
\|W\|_\infty \|\theta\|_1^2 
 \leq  \|W\|_\infty \|\theta\|_1
\leq  \lambda_1 \|\theta\|_1.
\end{eqnarray}
%
%
%
%
%
To treat the complementary set, $A_0^c$, we use the peeling device (\cite{sarabook}). Let $M_j:= 2^{j}$ and let $J$ be the smallest integer such that $2^{J} \geq \tune$. 
Consider partitioning of the set $A_0^c$ 
$$A_0^c = \bigcup_{j=1}^J A_j$$
where 
$$
A_j:=
\{\theta\in A: M_{j-1}\leq \|\theta\|_1 \leq M_j\}.$$

\noindent
Using the union bound and the definition of $A_j$ we obtain the sequence of upper bounds in the display below. 
Note that in the inequality \eqref{delta2} below, we used that $4\delta_{{M_{j-1}^2}} \geq \delta_{{M_{j}^2}}.$ 

\begin{eqnarray}
&&
P\left(  {\exists\; \theta\in\; \bigcup_{j=1}^J A_j} :\;
{|\theta^TW \theta|}{ } 
\geq  4\times 27\delta_{\|\theta\|_1^2}\|\theta\|_2^2
\right)
\\\nonumber
&& \quad\quad\leq
\sum_{j=1}^J P\left(  \exists\; {\theta\in A_j} :
{|\theta^TW \theta|}{  } 
\geq 4\times 27\delta_{\|\theta\|_1^2}\|\theta\|_2^2
\right)
\\\nonumber
&& \quad\quad\leq
\sum_{j=1}^J P\left(  \exists\;\theta\in A_j, \|\theta\|_2= 1:
{|\theta^TW \theta|} 
\geq 
4\times 27 \delta_{{M_{j-1}^2}} \|\theta\|_2^2
\right)
\\\label{delta2}
&& \quad\quad{\leq}
\sum_{j=1}^J P\left(   \exists\;\theta\in A_j ,\|\theta\|_2 = 1:
{|\theta^TW \theta|} 
\geq 
 27\delta_{{M_{j}^2}}
\right)
\\\nonumber
&& \quad\quad\leq
\sum_{j=1}^J P\left(  \exists \; {\theta,\|\theta\|_2 = 1, \|\theta\|_1 \leq M_j} :
{|\theta^TW \theta|} 
\geq 
27\delta_{M_{j}^2} 
\right)
\\\nonumber
&& \quad\quad\stackrel{(\text{Lemma }\ref{l1tol0})}{\leq} 
\sum_{j=1}^J P\biggl(  \exists \; {\theta\in\text{cl}(\text{conv}(\mathbb B_2(3)\cap \mathbb B_0(M_j^2) ))} :\\
&&\hspace{4cm}
{|\theta^TW \theta|} 
\geq 
27\delta_{M_{j}^2} 
\biggr).
\end{eqnarray}
We now show that if 
\begin{eqnarray}\label{imp1}
\sup_{v:\|v\|_2=1, \|v\|_0 \leq M}|v^T W v| \leq \delta,
\end{eqnarray}
then
\begin{eqnarray}
\label{imp2}
|\theta^T W\theta | \leq 27\delta,\quad\quad\forall
\theta\in\text{cl}(\text{conv}(\mathbb B_2(3)\cap \mathbb B_0(M) )).
\end{eqnarray}
First if $\theta\in\text{conv}(\mathbb B_2(3)\cap \mathbb B_0(M_j^2) )$, then we can write
$\theta = \sum_{i} v_i\alpha_i$, where $v_i\in \mathbb B_2(3) \cap \mathbb B_0(M_j^2)$.
For each $i,j$ it holds
\begin{eqnarray*}
|v_i^T W v_j | &=& \frac{1}{2}|(v_i+v_j)^T W (v_i+v_j) - v_i^T W v_i - v_j^T W v_j| \\
&\leq & 
\frac{1}{2} (36\delta+9\delta+9\delta)=27\delta.
\end{eqnarray*}
Hence
$$|\theta^T W\theta | = |\sum_{i,j}v_i^TWv_j | \leq \sum_{i,j}27\delta \alpha_i \alpha_j = 27\delta.$$
If $\theta$ is in the closure of the set $\text{conv}(\mathbb B_2(3)\cap \mathbb B_0(M_j^2) )$, we can obtain an analogous implication as \eqref{imp1} $\Rightarrow$ \eqref{imp2} by continuity arguments.
Therefore, we can continue the chain of bounds
\begin{eqnarray}
&&\sum_{j=1}^J P\left(  \exists \; {\theta\in\text{cl}(\text{conv}(\mathbb B_2(3)\cap \mathbb B_0(M_j^2) ))} :
{|\theta^TW \theta|} 
\geq 
27\delta_{M_{j}^2} 
\right)
\\\nonumber
&&\leq \sum_{j=1}^J P\left(  \sup_{v:\|v\|_2=1, \|v\|_0\leq M_j} |v^T Wv| \geq \delta_{M_j^2}
\right)
\\\nonumber
&&
\stackrel{(\text{Lemma }\ref{dev})}{\leq}
\sum_{j=1}^J 2e^{-{M_{j}^2}\log (2p)}
\\\label{cup}
&& \stackrel{\text{since }M_j\geq 1}{
\leq } 
2J e^{- \log (2p)}.
\end{eqnarray}
Therefore we conclude from \eqref{A0} and \eqref{cup} that
\begin{eqnarray*}
&&P(\exists \;\theta\in A\;:|\theta^TW \theta| 
\geq  \lambda_1\|\theta\|_1 + 4\delta_{\|\theta\|_1^2}\|\theta\|_2^2 )
\\[0.1cm]
&&
\;\; \leq 
P(\exists \;\theta\in A_0: |\theta^TW \theta|  \geq  \lambda_1\|\theta\|_1)
\\
&&\;\;\;\;+\;
P(\exists \;\theta\in A_0^c\;:|\theta^TW \theta| 
\geq  4\delta_{\|\theta\|_1^2}\|\theta\|_2^2 )
\\[0.1cm]
&&\;\;
\leq
2(J+1)e^{-\log(2p)} 
.
\end{eqnarray*}


\end{proof}

\bibliography{../../../gminf}

\begin{thebibliography}{}

\bibitem[{Amini} and {Wainwright}, 2009]{amini.wainwright}
{Amini}, A. and {Wainwright}, M. (2009).
\newblock High-dimensional analysis of semidefinite relaxations for sparse
  principal components.
\newblock {\em Annals of Statistics}, 37(5b):2877--2921.

\bibitem[Anderson, 1963]{andersen}
Anderson, T.~W. (1963).
\newblock Asymptotic theory for principal component analysis.
\newblock {\em Annals of Mathematical Statistics}, 34(1):122--148.

\bibitem[{Bai} and {Yin}, 1993]{bai.yin}
{Bai}, Z.~D. and {Yin}, Y.~Q. (1993).
\newblock Limit of the smallest eigenvalue of large dimensional covariance.
\newblock {\em Annals of Probability}, 21(3):1275--1294.

\bibitem[Baik and Silverstein, 2006]{baiksilverstein}
Baik, J. and Silverstein, J.~W. (2006).
\newblock Eigenvalues of large sample covariance matrices of spiked population
  models.
\newblock {\em Journal of Multivariate Analysis}, 97:1382--1408.

\bibitem[{Belloni} et~al., 2015]{vch}
{Belloni}, A., {Chernozhukov}, V., and {Kato}, K. (2015).
\newblock {Uniform post selection inference for LAD regression and other
  Z-estimation problems}.
\newblock {\em Biometrika}, 102(1):77--94.

\bibitem[Berthet and Rigollet, 2013]{berthet2013optimal}
Berthet, Q. and Rigollet, P. (2013).
\newblock Optimal detection of sparse principal components in high dimension.
\newblock {\em Annals of Statistics}, 41(4):1780--1815.

\bibitem[Birnbaum et~al., 2013]{birnbaum}
Birnbaum, A., Johnstone, I.~M., Nadler, B., and Paul, D. (2013).
\newblock Minimax bounds for sparse pca with noisy high-dimensional data.
\newblock {\em The Annals of Sta- tistics}, 41:1055--1084.

\bibitem[{B\"uhl\-mann} and {van de Geer}, 2011]{hds}
{B\"uhl\-mann}, P. and {van de Geer}, S. (2011).
\newblock Statistics for high-dimensional data.
\newblock {\em Springer}.

\bibitem[{Cai} and {Guo}, 2015]{cai.guo}
{Cai}, T. and {Guo}, Z. (2015).
\newblock Confidence intervals for high-dimensional linear regression: Minimax
  rates and adaptivity.
\newblock {\em ArXiv: 1506.05539}.

\bibitem[Cai et~al., 2013]{cai.spca}
Cai, T., Ma, Z., and Wu, Y. (2013).
\newblock Sparse {P}{C}{A}: Optimal rates and adaptive estimation.
\newblock {\em Annals of Statistics}, 41(6):3074--3110.

\bibitem[Chernozhukov et~al., 2015]{vch1}
Chernozhukov, V., Hansen, C., and Spindler, M. (2015).
\newblock Valid post-selection and post-regularization inference: An
  elementary, general approach.
\newblock {\em Annual Review of Economics}, 7(1):649--688.

\bibitem[{d'Aspremont} et~al., 2007]{aspremont.direct}
{d'Aspremont}, A., {El Ghaoui}, L., Jordan, M., and Lanckriet, G. (2007).
\newblock {A Direct Formulation for Sparse PCA Using Semidefinite Programming}.
\newblock {\em SIAM Review}, 49(3):434--448.

\bibitem[Deshpande and Montanari, 2014]{montanari.pca}
Deshpande, Y. and Montanari, A. (2014).
\newblock Sparse {P}{C}{A} via covariance thresholding.
\newblock In {\em Advances in Neural Information Processing Systems}, pages
  334--342.

\bibitem[{Fan} and {Wang}, 2015]{fan.spca}
{Fan}, J. and {Wang}, W. (2015).
\newblock {Asymptotics of Empirical Eigen-structure for Ultra-high Dimensional
  Spiked Covariance Model}.
\newblock {\em ArXiv:1502.04733}.

\bibitem[{Jankov\'a} and {van de Geer}, 2015]{jvdgeer14}
{Jankov\'a}, J. and {van de Geer}, S. (2015).
\newblock Confidence intervals for high-dimensional inverse covariance
  estimation.
\newblock {\em Electronic Journal of Statistics}, 9(1):1205 --1229.

\bibitem[Jankov\'a and {van de Geer}, 2016]{anor}
Jankov\'a, J. and {van de Geer}, S. (2016).
\newblock Confidence regions for generalized linear models under sparsity.
\newblock {\em ArXiv: 1610.01353}.

\bibitem[{Jankov\'a} and {van de Geer}, 2016]{jvdgeer15}
{Jankov\'a}, J. and {van de Geer}, S. (2016).
\newblock Honest confidence regions and optimality for high-dimensional
  precision matrix estimation.
\newblock {\em TEST}, 26(1):143--162.

\bibitem[{Javanmard} and {Montanari}, 2014]{stanford1}
{Javanmard}, A. and {Montanari}, A. (2014).
\newblock {Confidence intervals and hypothesis testing for high-dimensional
  regression}.
\newblock {\em Journal of Machine Learning Research}, 15(1):2869--2909.

\bibitem[Johnstone, 2001]{johnstone2}
Johnstone, I.~M. (2001).
\newblock On the distribution of the largest eigenvalue in principal components
  analysis.
\newblock {\em Annals of Statistics}, 29(2):295--327.

\bibitem[Johnstone and Lu, 2009]{jl}
Johnstone, I.~M. and Lu, A.~Y. (2009).
\newblock On consistency and sparsity for principal components analysis in high
  dimensions.
\newblock {\em Journal of the American Statistical Association},
  104(486):682--693.

\bibitem[Jolliffe et~al., 2003]{jolliffe2003modified}
Jolliffe, I.~T., Trendafilov, N.~T., and Uddin, M. (2003).
\newblock A modified principal component technique based on the lasso.
\newblock {\em Journal of Computational and Graphical Statistics},
  12(3):531--547.

\bibitem[Kollo and Neudecker, 1997]{var.PC}
Kollo, T. and Neudecker, H. (1997).
\newblock Asymptotics of {P}earson-{H}otelling principal-component vectors of
  sample variance and correlation matrices.
\newblock {\em Behaviormetrika}, 24(1):51--69.

\bibitem[{Koltchinskii} et~al., 2017]{2017arXiv170807642K}
{Koltchinskii}, V., {L{\"o}ffler}, M., and {Nickl}, R. (2017).
\newblock {Efficient Estimation of Linear Functionals of Principal Components}.
\newblock {\em ArXiv e-prints}.

\bibitem[{Koltchinskii} and {Lounici}, 2017]{2016arXiv160101457K}
{Koltchinskii}, V. and {Lounici}, K. (2017).
\newblock {New asymptotic results in principal component analysis}.
\newblock {\em Sankhya A}, 79(254).

\bibitem[Koltchinskii et~al., 2016]{koltchinskii2016asymptotics}
Koltchinskii, V., Lounici, K., et~al. (2016).
\newblock Asymptotics and concentration bounds for bilinear forms of spectral
  projectors of sample covariance.
\newblock In {\em Annales de l'Institut Henri Poincar{\'e}, Probabilit{\'e}s et
  Statistiques}, volume~52, pages 1976--2013. Institut Henri Poincar{\'e}.

\bibitem[Koltchinskii et~al., 2017]{kolt.normal}
Koltchinskii, V., Lounici, K., et~al. (2017).
\newblock Normal approximation and concentration of spectral projectors of
  sample covariance.
\newblock {\em The Annals of Statistics}, 45(1):121--157.

\bibitem[{Loh} and Wainwright, 2014]{non-convex1}
{Loh}, P. and Wainwright, M. (2014).
\newblock Regularized {M}-estimators with nonconvexity: Statistical and
  algorithmic theory for local optima.
\newblock {\em Journal of Machine Learning Research}, 1:1--56.

\bibitem[{Loh} and {Wainwright}, 2012]{LohWai12}
{Loh}, P.-L. and {Wainwright}, M.~J. (2012).
\newblock {High-dimensional regression with noisy and missing data: Provable
  guarantees with nonconvexity}.
\newblock {\em Annals of Statistics}, 40(3):1637--1664.

\bibitem[Paul, 2007]{paul.pca}
Paul, D. (2007).
\newblock Asymptotics of sample eigenstructure for a large dimensional spiked
  covariance model.
\newblock {\em Statistica Sinica}, 17:1617--1642.

\bibitem[Shen et~al., 2013]{shen2013surprising}
Shen, D., Shen, H., Zhu, H., and Marron, J. (2013).
\newblock Surprising asymptotic conical structure in critical sample
  eigen-directions.
\newblock {\em ArXiv:1303.6171}.

\bibitem[{van de Geer}, 2000]{sarabook}
{van de Geer}, S. (2000).
\newblock {\em {Empirical processes in M-estimation}}.
\newblock Springer.

\bibitem[van~de Geer, 2014]{uniform}
van~de Geer, S. (2014).
\newblock On the uniform convergence of empirical norms and inner products,
  with application to causal inference.
\newblock {\em Electronic Journal of Statistics}, 8(1):543--574.

\bibitem[{van de Geer}, 2016]{sf}
{van de Geer}, S. (2016).
\newblock {\em Estimation and Testing under Sparsity: \'Ecole d'\'Et\'e de
  Saint-Flour XLV}.
\newblock Springer.

\bibitem[{van de Geer} et~al., 2014]{vdgeer13}
{van de Geer}, S., {B{\"u}hlmann}, P., {Ritov}, Y., and {Dezeure}, R. (2014).
\newblock {On asymptotically optimal confidence regions and tests for
  high-dimensional models}.
\newblock {\em Annals of Statistics}, 42(3):1166--1202.

\bibitem[{Vu} et~al., 2013]{fantope}
{Vu}, V., {Cho}, J., {Lei}, J., and {Rohe}, K. (2013).
\newblock {Fantope Projection and Selection: A near-optimal convex relaxation
  of Sparse PCA}.
\newblock {\em Advances in Neural Information Processing Systems (NIPS)}, 26.

\bibitem[Vu and Lei, 2012]{vu.lei}
Vu, V. and Lei, J. (2012).
\newblock {Minimax rates of estimation for sparse PCA in high dimensions}.
\newblock {\em Journal of Machine Learning Research}, 22:1278--1286.

\bibitem[{Zhang} and {Zhang}, 2014]{zhang}
{Zhang}, C.-H. and {Zhang}, S.~S. (2014).
\newblock {Confidence intervals for low-dimensional parameters in
  high-dimensional linear models}.
\newblock {\em Journal of the Royal Statistical Society: Series B},
  76:217--242.

\bibitem[Zou et~al., 2006]{zou.pca}
Zou, H., Hastie, T., and Tibshirani, R. (2006).
\newblock Sparse principal component analysis.
\newblock {\em Journal of Computational and Graphical Statistics}, 15:265--286.

\end{thebibliography}
\end{document}